\documentclass[11pt]{article}
\usepackage{style}
\graphicspath{ {./images/} }

\usepackage{blindtext}
\usepackage{algorithm,algpseudocode}
\usepackage{graphicx}
\usepackage{natbib}
\usepackage{subfiles}
\usepackage{cleveref}
\hypersetup{
  colorlinks   = true,
  urlcolor     = blue,
  linkcolor    = blue,
  citecolor   = blue
}
\usepackage{dsfont}
\usepackage{mathtools}
\usepackage{multirow}
\usepackage{authblk}
\usepackage{setspace}[]

\title{Evaluating Black-Box Classifiers via Stable Adaptive Two-Sample Inference}

\begin{document}
\author[1]{Yuchen Chen}
\author[1]{Jing Lei}
\affil[1]{Carnegie Mellon University}
\maketitle

\begin{abstract}
We consider the problem of evaluating black-box multi-class classifiers. In the standard setup, we observe class labels $Y\in \{0,1,\ldots,M-1\}$ generated according to the conditional distribution
$
    Y|X \sim \Cat\big(\eta(X)\big),
$
where $X$ denotes the features and $\eta$ maps from the feature space to the $(M-1)$-dimensional simplex. A black-box classifier is an estimate $\hat{\eta}$ for which we make no assumptions about the training algorithm. Given holdout data, our goal is to evaluate the performance of the classifier $\hat{\eta}$. Recent work suggests treating this as a goodness-of-fit problem by testing the hypothesis
    $H_0: \rho((X,Y),(X',Y')) \le \delta$,
where $\rho$ is some metric between two distributions, and $(X',Y')\sim P_X\times \Cat(\hat\eta(X))$. Combining ideas from algorithmic fairness, Neyman-Pearson lemma, and conformal p-values, we propose a new methodology for this testing problem. The key idea is to generate a second sample $(X',Y') \sim P_X \times \Cat\big(\hat\eta(X)\big)$ allowing us to reduce the task to two-sample conditional distribution testing. Using part of the data, we train an auxiliary binary classifier called a distinguisher to attempt to distinguish between the two samples. The distinguisher's ability to differentiate samples, measured using a rank-sum statistic, is then used to assess the difference between $\hat{\eta}$ and $\eta$ . Using techniques from cross-validation central limit theorems, we derive an asymptotically rigorous test under suitable stability conditions of the distinguisher.
\end{abstract}

\setstretch{1}

\section{Introduction}
Classification is a fundamental task in statistics and machine learning. 
In this work, we focus on evaluating black-box probabilistic multi-class classifiers. In the general set-up, we have features (covariates) $X \sim P_X$ supported on a feature space $\mathcal{X}$.  We also have a label space $\mathcal{L}$ consisting of $M$ labels $y=0,..,M-1$. The labels are generated by $Y \sim \Cat(\eta(X))$, where $\eta: \mathcal{X} \rightarrow \Delta_{M-1}$ is an unknown function mapping a feature vector $X$ to the probability simplex over $\mathcal{L}$ denoted by $\Delta_{M-1}$. Given a training sample of $(X,Y)$ pairs from the underlying distribution, the classification task is to train a probabilistic classifier $\hat{\eta}: \mathcal{X} \rightarrow \Delta_{M-1}$, an estimate of $\eta$.

Practitioners have access to a diverse class of rich and flexible classification methods at their disposal such as random forests, XgBoost, neural networks, etc.  In this work we consider a generic estimate $\hat\eta$, and do not make any assumptions about the inner workings of the training algorithm used to obtain $\hat\eta$, which we call a ``black-box'' estimate. After obtaining such a black-box estimate, we naturally want to know how good the classifier is. The importance of this question is further emphasized by the emergence of deep learning methods and the use of cloud-based services. These methods may be quite expensive to train. In an effort to save resources, one may wish to first train a classifier on a small subset of data and then test whether this classifier is sufficient or if more resources, such as more powerful algorithms and/or more training samples, are needed for the training process.  

A natural approach to this question is to evaluate the classifier's test accuracy on a hold-out set. The main drawback of using classification accuracy is that this metric does not quantify whether the estimated model reflects the data generating distribution. For example, consider the following binary classification example given in \citet{zhangClassificationProcedureGood2023}. Suppose that $M=2$ and $\eta(x) = (\frac{1}{2},\frac{1}{2})$ for all $x \in \mathcal{X}$. The classification accuracy of the probabilistic classifier $\hat{\eta}(x) \equiv (\frac{1}{2},\frac{1}{2})$ would be around $0.5$, which may seem poor,  but this classifier matches the data generating distribution, i.e. $\hat{\eta} = \eta$. Meanwhile, when the two labels are balanced, the probabilistic classifier $\hat{\eta}(x)\equiv(1,0)$ has the same classification accuracy, but the estimated parameter $\hat\eta(\cdot)$ is very different from the data generating distribution. In more complex settings, relying on test accuracy may lead to models with poor generalization. For instance, classification accuracy may be high because of overfitting \citep{zhangClassificationProcedureGood2023,javanmardGRASPGoodnessoffitTest2024}.

To combat the issues with using test accuracy, recent work has suggested evaluating black-box classifiers using Goodness-of-Fit (GoF) tests \citep{zhangClassificationProcedureGood2023,javanmardGRASPGoodnessoffitTest2024}. In the simplest form, given $(X_i,Y_i)_{i=1}^n\stackrel{iid}{\sim} P_X\times \Cat(\eta(X))$, we want to test 
\begin{equation}
\label{eq: exact hypothesis}
    H_0: \eta(X)=\hat{\eta}(X)\,,~~\text{almost surely}.
\end{equation}
where $\hat\eta(\cdot)$ is an external estimate of $\eta(\cdot)$ independent of $(X_i,Y_i)_{i=1}^n$. Two-sample tests are often used for such hypotheses.  
Suppose for now that we have additional $X$-marginal data $X'_1,...,X'_n\overset{iid}{\sim} P_X$. We can form a second sample $\{(X_i',Y_i')\}$ where $Y'_i \sim \Cat(\hat{\eta}(X'_i))$.

We propose using the following two-sample test which is related to previous works on algorithmic fairness and outcome indistinguishability \citep{dworkOutcomeIndistinguishability2021} and classification accuracy testing \citep{kimClassificationAccuracyProxy2021}. Label the observations in the first sample $(X_1,Y_1),...,(X_n,Y_n)$ with class label $C=0$ and label the observations in the second sample $(X'_1,Y'_1),...,(X'_n,Y'_n)$ with class label $C=1$. Using part of the data (index set $\mathcal{I}\subset[n]$), we train a distinguishing function $\hat g:\mathcal X\times\mathcal L\mapsto \mathbb R$, which is typically associated with a probabilistic classifier estimating $P(C=1|X,Y)$. We evaluate the test statistic
$$
    T = \frac{1}{|\mathcal{I}^c|^2} \sum_{i,j \in \mathcal{I}^c} \mathds{1}\left(\hat g(X_i,Y_i) < \hat g(X'_j,Y'_j)\right).
$$
The intuition behind this test statistic is that under the null, the two samples should be ``indistinguishable", so the distinguisher $\hat g$ cannot tell the sample generated using $\hat{\eta}$ apart from the sample generated using $\eta$. In this case $T \approx 0.5$ as $\hat g$ cannot do any better than random guessing. When there is separation between $P_X \times \Cat(\hat{\eta}(X))$ and $P_X \times \Cat(\eta(X))$, $\hat g$ should be significantly larger on the second sample compared to the first, so $T > 0.5$. We emphasize that even though we propose to evaluate a classifier using another classifier---the distinguisher---the estimation of it is usually significantly easier than the original classification task. For instance, estimating the distinguisher is always a binary classification problem, which has a latent low-dimensional structure since we know that the $X$-marginals are equal, the only signal is the conditional distribution of $Y|X$. Moreover, the distinguisher does not need to be perfect, it only needs to detect some difference between the two samples to have nontrivial power. Previous works using similar test procedures show that this statistic is still powerful even when the distinguisher is not consistent \citep{caiAsymptoticDistributionFreeIndependence2024a}. Conditional on the estimated distinguisher $\hat g$, $T$ is a two-sample U-statistic that satisfies a central limit theorem \citep{vaartAsymptoticStatistics1998a}. This normal approximation can be used to construct asymptotically valid rejection cutoffs.

Several additional extensions are needed to make the proposed procedure practical. The above procedure requires additional $X$-marginal data to form the second sample. A priori, we do not have access to such data. To actually implement this test procedure, we will need to consider how to construct the second sample used in the two-sample test. In addition, we do not want to test the exact equality in hypothesis \eqref{eq: exact hypothesis} because in practice we would never expect an estimate to be perfect, but rather we want to test whether the two distributions are ``close'' in some sense. We formalize this notion as a tolerant testing problem
\begin{equation}
\label{eq: tol hypothesis}
    H_0: \rho\bigg(P_X \times \Cat\big(\hat{\eta}(X)\big), P_X \times \Cat\big(\eta(X)\big)\bigg) < \delta,
\end{equation}
where $\rho$ is some measure of separation between the distributions. We will discuss how to construct the second sample and the extension to the tolerant hypothesis in detail when we describe the complete procedure in Section \ref{sec: procedure}. Additionally, we use cross-fitting to avoid loss of sample size efficiency caused by sample splitting.

We highlight some key features of our proposed procedure. Our method requires minimal assumptions on the distribution $(X,Y)$. In addition, we make no structural assumptions on the classifier $\hat{\eta}$, only assuming access to $\hat{\eta}$ through sampling. However, our method is adaptive to additional distributional information by leveraging the distinguisher $\hat g$. For example, if we expect that the signal in the $X$-marginal is sparse, we can use methods tailored to sparse classification (e.g. Lasso) to construct $\hat g$, allowing our test to be effective in high-dimensional and other complex settings. Finally, to the best of our knowledge, our test is the only one with a rigorous guarantee in the multi-class setting.

\subsection*{Related Work}
\paragraph*{Outcome indistinguishability}
Our strategy is closely related to the recent notion of outcome indistinguishability in the algorithmic fairness literature \citep{dworkOutcomeIndistinguishability2021,dwork2022beyond}. In essence, outcome indistinguishability posits that the outcomes generated by an ideal classifier cannot be distinguished from samples generated by nature. Our test can be viewed as testing whether the classifier given by $\hat{\eta}$ is outcome indistinguishable with respect to a distinguisher trained using the holdout samples. \citep{dworkOutcomeIndistinguishability2021} proposes several different levels of outcome indistinguishability. We describe the most basic form. Consider the binary classification setting where the number of labels $M=2$ and we view $\eta: \mathcal{X} \rightarrow [0,1]$ where $\eta(x) = \P[Y=1|X=x]$. For  a given class of distinguishers $\mathcal{A}$ and $\epsilon>0$, a probabilistic classifier $\hat{\eta}$ is $(\mathcal{A},\epsilon)$ outcome indistinguishable if 
$
    \Big|\P[A(X,Y) = 1] - \P[A(X,Y') = 1]\Big| \leq \epsilon,
$
for every $A\in\mathcal{A}$ and $X \sim P_X, Y\sim \Bern(\eta(X)), Y'\sim \Bern(\hat{\eta}(X'))$. Thus, outcome indistinguishibility says a classifier is ``good'' if any distinguisher in $\mathcal{A}$ cannot tell, up to a tolerance level $\epsilon$, the difference between true nature generated samples and samples generated from $\hat{\eta}$. In our setting, we choose $\mathcal{A}$ to be the class of all probabilistic classifiers trained on the holdout data. One difference with our approach is that we consider a rank-based statistic that can relate the tolerance $\epsilon$ to the area under the ROC curve of the distinguishers. We also extend the idea of indistinguishability to multi-class problems.

\paragraph*{GoF tests for binary classification}
There is a long history of goodness-of-fit tests for parametric models. These methods mostly focus on variants of Pearson's chi-square statistic. One line (of many) of work in this area is \citet{mccullaghAsymptoticDistributionPearsons1985,osiusNormalGoodnessofFitTests1992,farringtonAssessingGoodnessFit1996}. More recent work has studied goodness-of-fit tests for generalized linear models in high dimensional settings \citep{shah2018goodness,jankova2020goodness}. With the development and popularity of new nonparametric classification methods, there is a need to extend to these nonparametric settings.

Recently, there have been two developments in the nonparametric setting. \citet{zhangClassificationProcedureGood2023} proposed a binary adaptive goodness-of-fit test called BAGofT. The main idea is to use the estimated classifier to form an adaptive binning in the covariate space which turns the problem into testing multiple binomials, and then use Pearson's chi-square statistic. They aim to test the hypothesis
$
    H_0: \sup_x |\hat{\eta}(x) - \eta(x)| = O_p(r_n),
$
where $r_n$ is the convergence rate of the estimated classifier under the null. There are a few fundamental differences in this setup compared to ours. First, they test whether the estimated classifier is asymptotically equal to the true classifier while we test whether the estimated classifier is in a user-specified radius of the true classifier. Second, their procedure and theory does not treat the estimated classifier $\hat{\eta}$ as a black-box in that the training of $\hat{\eta}$ is a part of the test procedure. Finally, BAGofT only tests binary classification methods. \citet{zhangClassificationProcedureGood2023} offers a multi-class extension, but it is not rigorously justified and may be computationally heavy.

\citet{javanmardGRASPGoodnessoffitTest2024} proposes a randomization test to test the hypothesis
\begin{equation*}
    H_0: \rho\bigg(P_X \times \Bern\big(\hat{\eta}(X)\big), P_X \times \Bern\big(\eta(X)\big)\bigg) < \delta,
\end{equation*}
where $\rho$ is an $f$-divergence, such as total variation, chi-square divergence, etc. Thus, the setting considered by \citet{javanmardGRASPGoodnessoffitTest2024} is more comparable to ours than the setting considered in \citet{zhangClassificationProcedureGood2023}. Their method, called GRASP, is to form a randomization test which follows three main steps. First, for each sample, they generate $m$ counterfeit samples using the black-box classifier $\hat{\eta}$. The second step involves ranking the true sample among the counterfeits. From these ranks, a decision rule is computed using an optimization subroutine. Although both of our test procedures involve constructing multiple samples and ranks, they are done in different ways. This difference allows for some benefits compared to the GRASP method. First, our method allows for an easy extension to the multiclass classifier case. Furthermore, constructing the decision rule using normal approximation allows us to avoid the additional computational burden of the optimization subroutine and enables convenient construction of one-sided confidence intervals for the parameter of interest: the difference between the true distribution and the black-box output.

\paragraph*{Calibration}
\citet{leeTCalOptimalTest2023} presents a similar method of constructing a second sample and using a two-sample test to evaluate whether a classifier is calibrated. As pointed out by \citet{javanmardGRASPGoodnessoffitTest2024}, the test objective is different, in that in the $\delta=0$ case, calibration aims to test whether
$
    \E[\eta(X)|\hat{\eta}=\eta] = \eta,
$
while we aim to test whether
$
    \eta(X) = \hat{\eta}(X)~a.s.
$
Because of the differences in null hypothesis, the actual samples used in the two-sample tests differ between our work and \citet{leeTCalOptimalTest2023}. Moreover, we emphasize that the test procedure of \citet{leeTCalOptimalTest2023} leverages distributional assumptions such as smoothness, while our test does not make distributional assumptions.

\paragraph*{Two-Sample testing by classification}
The idea of using classification accuracy in two-sample testing has been studied recently \citep{kimClassificationAccuracyProxy2021,gerberMinimaxOptimalTesting2023}.  Our test statistic is slightly different as we use a rank-based approach. This type of rank-based approach has been applied for other problems such as independence testing \citep{caiAsymptoticDistributionFreeIndependence2024a}. To the best of our knowledge, this work is the first to use cross-fitting for two-sample testing by classification, which involves adapting recent analysis of cross-validation central limit theorems and stability. The null hypothesis \eqref{eq: exact hypothesis} is a special case of the two-sample conditional distribution test considered in \cite{huTwoSampleConditionalDistribution2024a,chen2025biased}, who considered testing the equality of the conditional distribution of $Y$ given $X$ between two populations, and allowed general $Y$.  The ranking-based test statistic is inspired by this work but the tolerance testing formulation \eqref{eq: tol hypothesis} and the construction of the second sample are new.

\section{Test Procedure}
\label{sec: procedure}
Suppose we are in a multi-class classification setting with label $Y$ taking values $0,...,M-1$, where $M>1$ is the total number of labels and features $X$ with distribution $P_X$ supported on a feature space $\mathcal{X}$. The underlying data-generating distribution has the form $ (X, Y) \sim P_X \times \Cat\big(\eta(X)\big),
$
where $\eta(x) \in \Delta_{M-1}$ is the conditional probability distribution of the label $Y$ conditional on $X=x$. Given a probabilistic classifier $\hat{\eta}:\mathcal{X} \rightarrow \Delta_{M-1}$ and a hold-out dataset $(X_1,Y_1),...,(X_n,Y_n)$ independent of the data used to train $\hat{\eta}$, we want to test the hypothesis
\begin{equation*}
    H_0: \rho\Big(P_X \times \Cat\big(\hat{\eta}(X)\big), P_X \times \Cat\big(\eta(X)\big)\Big) < \delta,
\end{equation*}
where $\rho$ is a metric to be specified later and $\delta>0$ is a pre-specified tolerance.

Recall that we propose the following procedure. Given two samples $\{(X_i,Y_i)\} \sim P_X \times \Cat\big(\eta(X)\big)$ and $\{(X'_i,Y'_i)\} \sim P_X \times \Cat\big(\hat{\eta}(X)\big)$, we train a distinguisher $\hat g$, which estimates the conditional probability that a sample belongs to $P_X \times \Cat\big(\eta(X)\big)$ or $P_X \times \Cat\big(\hat{\eta}(X)\big)$ and then evaluate a rank-sum test statistic.

To formulate a practical two-sample test, we need to answer three questions.
\begin{enumerate}
    \item How do we generate the second sample?
    \item What measure $\rho$ do we use?
    \item How do we handle possible double dipping in estimation of the distinguisher $\hat g$ and evaluation of the rank-sum test statistic?
\end{enumerate}

\subsection{Sample augmentation and distinguishing}

For a two-sample test, we need a sample from $P_X \times \Cat\big(\hat{\eta}(X)\big)$. The simplest way to obtain this second sample is to use a further sample split.
We take an alternate approach that does not require this additional sample split.
Using the data $(X_1,Y_1),...,(X_n,Y_n)$, we can artificially create a second sample $(X'_1,Y'_1),...,(X'_n,Y'_n)$ where for $i=1,...,n$,
\begin{align}
    \begin{split}
    \label{eq: second sample}
    X'_i = X_i, \quad 
    Y'_i\mid X_i' \sim \Cat\big(\hat{\eta}(X'_i)\big).
    \end{split}
\end{align}
Then $(X'_i,Y'_i) \sim P_X \times \Cat\big(\hat{\eta}(X)\big)$. For notation, we use
$
    \mathcal{D}_n := \{(X_i,Y_i,Y_i')\}_{i=1}^n,
$
to denote this augmented dataset. Using the augmented dataset, we can estimate the distinguisher using a symmetric classification procedure
\begin{equation}
\label{eq:distinguisher}
\wh{D}:\bigcup_{n=1}^\infty(\mathcal{X}\times \mathcal{L}\times \mathcal{L})^n \rightarrow \mathcal{G},
\end{equation}
where $\mathcal{G}$ is the collection of measurable functions from $\mathcal{X}\times\mathcal{L}$ to $[0,1]$. Here is how the procedure $\wh{D}$ works. Given $\mathcal{D}_n = \{(X_i,Y_i,Y'_i)\}_{i=1}^n$, we first form an expanded data set 
$$
    \Tilde{\mathcal{D}}_n = \{(X_1,Y_1,0),...,(X_n,Y_n,0),(X_1,Y_1',1),...,(X_n,Y_n',1)\}.
$$
This data set simply expands $\mathcal{D}_n$ to explicitly show the two-sample structure, where the last variable, denoted by $C={0,1}$, denotes whether the sample belongs to the first or second sample. Then $\hat g =\wh{D}(\mathcal{D}_n)$ is a measurable function from $\mathcal X\times\mathcal L$ to $[0,1]$, with $\hat g(x,y)$ estimating the conditional probability $P(C=1|X=x,Y=y)$ using data $\tilde{\mathcal{D}}_n$ as if $\tilde{\mathcal D}_n$ consists of independent samples from the equal-weight mixture of $P_X\times \Cat(\eta(X))$ and $P_X\times\Cat(\hat\eta(X))$. This can be done using any off-the-shelf classification procedure.  In Section \ref{sec: coupled sample}, we give evidence that we can effectively estimate the distinguisher under the dependency structure of the augmented dataset. To avoid confusion, we will call $\wh D$ the ``distinguisher procedure'' and we refer to any output $\hat g$ of $\wh D$ as a ``distinguisher''.

\subsection{Neyman-Pearson separation}
We have two distributions $P_0, P_1$ supported on a common space $\mathcal{Z}$ and want to construct a separation measure $\rho$, between $P_0$ and $P_1.$ Following the intuition of two-sample testing by classification, we should let $\rho(P_0,P_1)$ measure how much an optimal distinguisher can distinguish between $P_0,P_1$. We first formalize how we can measure the distinguishing power of any distinguisher using the area under the receiver operating characteristic (ROC) curve.

Let $\hat{g}:\mathcal Z\mapsto [0,1]$ be the output of a distinguisher procedure for the distributions $P_0,P_1$ and $\wh{L}$ be the corresponding likelihood ratio estimate
\begin{equation*}
    \wh{L}(z) = \frac{\hat{g}(z)}{1-\hat{g}(z)}.
\end{equation*} Let $f_k,F_k$ denote the pdf/cdf for $\wh{L}$ under $P_k$ ($k=0,1$). For simplicity of discussion, we assume for now that $\wh L$ is a continuous random variable under both $P_0$ and $P_1$. The more general case can be handled using additional randomization, as described in the next subsection. We define the false positive rate (FPR) and the true positive rate (TPR) by
$$
    \FPR(t) = \int_{t}^\infty f_0 \qquad \TPR(t) = \int_{t}^\infty f_1.
$$
In the context of binary classification with labels $0,1$, we can construct a classifier by thresholding the corresponding likelihood ratio estimate $\wh{L}$ at a threshold $t\in [0,\infty)$. At a threshold $t$, the FPR represents the proportion of negatives that were classified as positives, and the TPR represents the proportion of positives that were classified as positives. A good distinguisher will minimize FPR and maximize TPR across the trajectory.

The tradeoff between FPR and TPR can be represented by the ROC curve, the curve $(\FPR(t),\TPR(t))$ parameterized by $t \in [0,\infty)$. The area under the ROC curve (AUC) can be used to aggregate the tradeoff across all thresholds into a single measure. A standard calculus computation shows that the area under the ROC curve (AUC) has the form
\begin{equation}
    \begin{split}
    \label{eq: AUC}
        \AUC(\hat g) &= \int_{0}^\infty (1-F_1(t))f_0(t) dt
        = \P(\wh{L}(Z)<\wh{L}(Z'))
        =\P(\hat{g}(Z)<\hat{g}(Z')),
    \end{split}
\end{equation}
where $(Z,Z')\sim P_0\times P_1$ and the last line is because $\hat{g}$ and $\wh{L}$ are related by a monotone transformation. AUC has been a metric used to evaluate binary classifiers. For a given $\wh{L}$, $\AUC(\wh{L})=0.5$ implies that $\wh{L}$ cannot distinguish between $P_0,P_1$ and $\AUC(\wh{L})=1$ implies that $\wh{L}$ can perfectly distinguish between the two distributions.

Suppose that $P_0,P_1$ have densities with respect to some common measure space. We can define the likelihood ratio $L:\mathcal{Z} \rightarrow [0,\infty]$
$
    L(z) = \frac{dP_1(z)}{dP_0(z)}.
$
The Neyman-Pearson Lemma \citet{neyman1933ix} shows that uniformly across all thresholds, the true likelihood ratio optimizes the FPR/TPR tradeoff. Consequently, the distinguisher that optimizes the AUC is given by the true likelihood ratio. 
  We now define the Neyman-Pearson metric characterizing the separation of $P_0,P_1$ as
$$
    \rho(P_0,P_1) = \AUC(L) - 0.5,
$$
where $L$ is the likelihood ratio. At the extremes, when $\rho(P_0,P_1) = 0$, no classifier can distinguish between the two distributions and when $\rho(P_0,P_1) = 0.5$, perfect classification is possible. This measure of separation is closely related to the total variation metric. The relationship is shown in the following proposition given in \citet{caiAsymptoticDistributionFreeIndependence2024a}(Proposition 2.1).
\begin{proposition}
    \label{prop: tv and auc}
    Let $\rho_{\rm tv}$ denote the total variation distance. Then we have
    $$
        \frac{1}{4}\rho_{\rm tv}(P_0,P_1) \leq \rho(P_0,P_1) \leq \frac{1}{2}\rho_{\rm tv}(P_0,P_1).
    $$
\end{proposition}

We remark that $\rho$ may not be a distance/metric as the triangle inequality may not be satisfied. However, Proposition \ref{prop: tv and auc} and the fact that $\rho_{\rm tv}$ satisfies the triangle inequality show that an approximate triangle inequality holds up to constants (i.e., it is a semi-metric).

If we have independent data $(Z_i)_{i=1}^n\sim P_0$, $(Z_j')_{j=1}^n\sim P_1$, (\ref{eq: AUC}) suggests that the rank-sum 
$$
    \frac{1}{n^2} \sum_{ij} \mathds{1}(\hat{g}(Z_i) < \hat{g}(Z'_j)),
$$
is a natural estimator of $\AUC(\hat{g}).$ Since Neyman-Pearson implies that $\AUC(\hat{g}) \leq \AUC(L) = \rho(P_0,P_1) + 0.5$, we can use the above rank-sum statistic based on any distinguisher $\hat g$ to infer a confidence lower bound for the separation $\rho$ which is sufficient to construct an asymptotically valid test for the hypothesis (\ref{eq: tol hypothesis}).  The asymptotic normality of the rank-sum statistic is formally established in Section \ref{sec: Theory} below. In the next subsection, we discuss dealing with the randomness in estimating the distinguisher $\hat g$.

\subsection{Avoid double-dipping}

If the same data is used for estimating the distinguisher and computing the rank sum statistic, we run into a double dipping issue which compromises the validity of the test procedure. We propose two test procedures below to avoid double dipping. The simplest involves sample-splitting and using the separate samples to estimate the distinguisher and evaluate the test statistic. To combat the loss of power from the reduced sample size of sample-splitting, we also propose a cross-fitting procedure which is valid whenever the distinguisher procedure $\widehat{D}$ is stable. Details on the asymptotics can be found in Section \ref{sec: Theory}.

\paragraph*{Sample-split procedure}
\label{sec: ss}
We first introduce a method using sample splitting. We partition $[n]$ into two partitions $\mathcal{I}_{n,1}$ and $\mathcal{I}_{n,2}$ with cardinality $n_1,n_2$. This forms sample splits $\mathcal{D}_{n,1} = \{(X_i,Y_i,Y'_i)\}_{i \in \mathcal{I}_{n,1}}$ and $\mathcal{D}_{n,2} = \{(X_i,Y_i,Y'_i)\}_{i \in \mathcal{I}_{n,2}}$. Let $\hat g_{n,1}$ be the output of the distinguisher procedure $\wh D$ in \eqref{eq:distinguisher} applied to the first split $\mathcal{D}_{n,1}$, i.e. $\hat g_{n,1} := \wh D(\mathcal{D}_{n,1})$. Using the other split $\mathcal{D}_{n,2}$, we can compute a rank sum statistic
\begin{equation*}
    T_{n,split} = \frac{1}{n_2^2} \sum_{i, j \in \mathcal{I}_2} R_{ij}(\mathcal{D}_{n,1}),
\end{equation*}
where we define $n_k = |\mathcal D_{n,k}|$ for $k=1,2$, and
\begin{equation}
\label{eq: rank ss}
    R_{ij}(\mathcal{D}_{n,1}) := \mathds{1}\bigg(\hat g_{n,1}(X_i,Y_i)<\hat g_{n,1}(X_j,Y'_j)\bigg).
\end{equation}

If $\hat{\eta} = \eta$, then $\hat g_{n,1}$ is unable to tell the difference between $(X_i,Y_i)_{i\in\mathcal{I}_{n,2}}$ and $(X_j,Y'_j)_{j\in\mathcal{I}_{n,2}}$, so we would expect $
    T_{n,split} \approx 0.5.
$ On the other hand, if $\wh g_{n,1}$ is able to distinguish the two samples, then on average $\hat g_{n,1}(X_j,Y'_j)$ should be larger than $\hat g_{n,1}(X_i,Y_i)$, so we would expect $T_{n,split} > 0.5$. Thus, we should reject for large values of $T_{n,split}$. To specify the rejection cutoff, we use normal approximation
$$
    \frac{\sqrt{n_2}}{\hat{\sigma}_{n,split}} \big[T_{n,split} - \mu(\mathcal{D}_{n,1})\big] \xrightarrow{d} \N(0,1),
$$
where $\mu(\mathcal D_{n,1})={\rm AUC}(\hat g_{n,1})$ and the scaling $\hat{\sigma}_{n,split}$ is defined in (\ref{eq: var est split}).
Although the centering $\mu(\mathcal{D}_{n,1})$ is a random quantity, we can still use this normal approximation to construct an asymptotically valid cutoff. We know that
$$
    \P\left(\frac{\sqrt{n_2}}{\hat{\sigma}_{n,split}} \big[T_{n,split} - \mu(\mathcal{D}_{n,1})\big] >z_{1-\alpha}\right) \leq \alpha+o(1).
$$
By Neyman-Pearson, we know that the center satisfies
$$
\mu(\mathcal{D}_{n,1})- \frac{1}{2} \leq \rho\bigg(P_X \times \Cat\big(\eta(X)\big), P_X \times \Cat\big(\hat{\eta}(X)\big)\bigg) < \delta,
$$
under the null. Thus, under the null
$$
    \P\left(\frac{\sqrt{n_2}}{\hat{\sigma}_{n,split}} \left[T_{n,split} - \delta-\frac{1}{2}\right] >z_{1-\alpha}\right) \leq \alpha+o(1).
$$
\begin{algorithm}[H]
\caption{Sample-split Procedure}
\label{alg:sample split}
\begin{algorithmic}[1]
\State Input: Data $(X_1,Y_1),...,(X_n,Y_n)$, probabilistic classifier $\hat{\eta}$, radius $\delta>0$, distinguisher procedure $\wh{D}$, nominal level $\alpha>0$.
\State Generate the second sample $(X'_1,Y'_1),...,(X'_n,Y'_n)$ according to (\ref{eq: second sample}).
\State Split the indices $[n]$ into two sets $\mathcal{I}_1$ and $\mathcal{I}_2$, form sample splits $\mathcal{D}_{n,1} = \{(X_i,Y_i,Y'_i)\}_{i \in \mathcal{I}_{n,1}}$ and $\mathcal{D}_{n,2} = \{(X_i,Y_i,Y'_i)\}_{i \in \mathcal{I}_{n,2}}$.
\State Using $\mathcal{D}_{n,1}$, fit the distinguisher $\hat{g}_{n,1} := \wh{D}(\mathcal{D}_{n,1})$.
\State Compute the test statistic $T_{n,split} = \frac{1}{n_2^2} \sum_{i, j \in \mathcal{I}_2} R_{ij}(\mathcal{D}_{n,1})$.
\State Reject if $
    \frac{\sqrt{n_2}(T_{n,split}-\delta-\frac{1}{2})}{\wh{\sigma}_{n,split}} > z_{1-\alpha},
$ where $\wh{\sigma}_{n,split}$ is defined in (\ref{eq: var est split}) and $z_{1-\alpha}$ is the corresponding normal quantile.
\end{algorithmic}
\end{algorithm}

In practice, we may have degenerate test statistics when $\hat g_{n,1}$ has point mass so that $\hat g_{n,1}(X_i,Y_i) = \hat g_{n,1}(X_j,Y'_j)$ with positive probability. In general, such degeneracy can be avoided using a random tie-breaking. Instead of using the ranks $R_{ij}$ defined in \eqref{eq: rank ss},
we first generate iid Uniform$([0,1])$ random variables $\{U_i\}_{i \in [n_{2}]}$ and $\{U'_j\}_{j\in [n_{2}]}$ and use
\begin{align}
    \widetilde{R}_{ij}(\mathcal{D}_{n,1}) := & \mathds{1}\bigg(\hat g_{n,1}(X_i,Y_i)<\hat g_{n,1}(X_j,Y'_j)\bigg)\nonumber\\
    ~~~&+ \mathds{1}\bigg(U_i < U'_j\bigg)\mathds{1}\bigg(\hat g_{n,1}(X_i,Y_i)=\hat g_{n,1}(X_j,Y'_j)\bigg)\,,\label{eq: random tie-breaking}
    \end{align}
instead. This external randomness allows us to effectively treat $\wh g_{n,1}$ as a continuous variable. For brevity, we will use the ranks $R_{ij}$ in the remainder of this paper. The extension to tie-breaking ranks $\widetilde{R}$ just involves additional bookkeeping.

\paragraph*{Cross-fitting procedure}

One major drawback of Algorithm \ref{alg:sample split} is that it requires sample splitting, which may negatively affect the empirical performance of this procedure by reducing the effective sample size. A common strategy to regain some of the lost efficiency in sample splitting is cross-fitting. In this section, we show how to incorporate cross-fitting into the proposed two sample test.

Recall $\mathcal{D}_n:=\{(X_i,Y_i,Y_i')\}_{i=1}^n$ is the augmented dataset and $\wh{D}:(\mathcal{X}\times\mathcal{L}\times\mathcal{L})^n \rightarrow \mathcal{G}$ is a distinguisher procedure. For the cross-fitting procedure with $K$ folds, let $\mathcal{I}_{n,1},...,\mathcal{I}_{n,K}$ be a partition of $[n]$. To simplify the presentation, we will assume that $|\mathcal{I}_{n,k}| = \frac{n}{K}$ for all $k=1,...,K$ and that $\frac{n}{K}$ is an integer. For $k=1,...,K$, we use $
    \mathcal{D}_{n,k} := \{(X_i,Y_i,Y_i')\}_{i \in \mathcal{I}_{n,k}}
$ to denote the data within the $k$-th fold and $\mathcal{D}_{n,-k} := \{(X_i,Y_i,Y_i')\}_{i \not\in \mathcal{I}_{n,k}}$ to denote the data outside of the $k$-th fold. For notation, let $\hat g_{n,-k}$ denote the output of the distinguisher procedure $\wh{D}$ applied to the data outside of the $k$-th fold, i.e. $\hat g_{n,-k} := \wh{D}(\mathcal{D}_{n,-k})$. Define
\begin{equation}
\label{eq: rank cf}
    R_{ij}(\mathcal{D}_{n,-k}) := \mathds{1}\bigg(\hat g_{n,-k}(X_i,Y_i)<\hat g_{n,-k}(X_j,Y'_j)\bigg).
\end{equation}

The cross-fitted version of the rank-sum test statistic is
$$
T_{n,cross} = \frac{1}{K}\sum_{k=1}^K\frac{1}{n_k^2}\sum_{i,j \in \mathcal{I}_k}R_{ij}(\mathcal{D}_{n,-k}),
$$
The major difficulty of the cross-fitted statistic is that the U-statistic kernel changes across the folds and so we should not expect a central limit theorem to hold in the usual way. However, we find that under some standard stability conditions on the distinguisher $\wh{D}$, the cross-fitted test statistic will satisfy, for some appropriately chosen scaling $\hat\sigma_{tr,cross}$,
$$
\frac{\sqrt{n}}{\hat{\sigma}_{tr,cross}}\left(T_{n,cross} - \frac{1}{K}\sum_{k=1}^K \mu(\mathcal{D}_{n,-k})\right) \xrightarrow{d} \N(0,1).
$$

The intuition is that if the estimating procedure is stable, the randomness of the U-statistic kernel is dominated by the sample points $(X_i,Y_i,Y_i')$, rather than the fitting procedure $\wh{D}$.
Similar to the sample split test, the centering is at a random value $\frac{1}{K}\sum_{k=1}^K \mu(\mathcal{D}_{n,-k})$ where
$$
    \mu(\mathcal{D}_{n,k}) - \frac{1}{2} = \AUC\left(\hat g_{n,-k}\right) - \frac{1}{2} \leq \rho\bigg(P_X \times \Cat(\eta(X)), P_X \times \Cat(\hat{\eta}(X))\bigg).
$$
Then under the null, we know that
$$
\frac{1}{K}\sum_{k=1}^K \mu(\mathcal{D}_{n,-k}) \leq \rho\bigg(P_X \times \Cat(\eta(X)), P_X \times \Cat(\hat{\eta}(X))\bigg)+ \frac{1}{2} < \delta + \frac{1}{2}.
$$
Then by similar argument as in the sample split test, we can reject when
$$
\frac{\sqrt{n}}{\hat{\sigma}_{tr,cross}}\left(T_{n,cross} - \delta-\frac{1}{2}\right) > z_{1-\alpha},
$$
where $\hat{\sigma}_{tr,cross}$ is defined in (\ref{eq: var est cross}).
The full cross-fit procedure is described in Algorithm \ref{alg:crossfit}.
\begin{algorithm}[H]
\caption{Cross-fit Procedure}
\label{alg:crossfit}
\begin{algorithmic}[1]
\State Input: Holdout data $(X_1,Y_1),...,(X_n,Y_n)$, probabilistic classifier $\hat{\eta}$, radius $\delta>0$, distinguishing procedure $\wh{D}$, number of splits $K$, nominal level $\alpha>0$.
\State Generate the second sample $(X'_1,Y'_1),...,(X'_n,Y'_n)$ according to (\ref{eq: second sample}).
\State Partition $[n]$ into $K$ folds $\mathcal{I}_{n,1},...,\mathcal{I}_{n,K}$.
\State On each fold $k=1,...,K$, compute the distinguisher $\hat g_{n,-k} := \wh{D}(\mathcal{D}_{n,-k})$.
\State Compute the test statistic $T_{n,cross} = \frac{1}{K}\sum_{k=1}^K\frac{1}{n_k^2}\sum_{i,j \in \mathcal{I}_k}R_{ij}(\mathcal{D}_{n,-k})$, where $R_{ij}$ is defined in (\ref{eq: rank cf}).
\State Reject if $
    \frac{\sqrt{n}(T_{n,cross}-\delta-\frac{1}{2})}{\hat{\sigma}_{tr,cross}} > z_{1-\alpha},
$ where $\hat{\sigma}_{tr,cross}$ is defined in (\ref{eq: var est cross}) and $z_{1-\alpha}$ is the corresponding normal quantile.
\end{algorithmic}
\end{algorithm}
\vspace{-1cm}
\paragraph*{Choice of $\delta$} Both the sample-split and cross-fit procedures, as stated in Algorithms (\ref{alg:sample split}) and (\ref{alg:crossfit}), output whether the hypothesis $H_0$ is rejected by a user prespecified choice of tolerance radius. However, a more interpretable procedure is to output $\delta_{\rm min}(\alpha)$, the smallest radius at which the test does not reject the null at level $\alpha$. By construction, we have
$
    \delta_{\rm min}(\alpha) = (T_{n,cross} -\frac{1}{2} - \frac{\hat{\sigma}_{tr,cross}}{\sqrt{n}}z_{1-\alpha})\vee 0\,,
$
which can be interpreted as a level $1-\alpha$ confidence lower bound of $\rho(P_X\times \eta\,,P_X\times \hat\eta)$.
A smaller $\delta_{\rm min}(\alpha)$ implies that $\hat{\eta}$ is a better classifier. This approach negates the need to prespecify a tolerance radius allowing the user to decide whether the $\delta_{\rm min}(\alpha)$ is sufficiently small for their purpose. We note that this approach is possible by using the normal approximation to derive the rejection cutoffs. In contrast, the use of an optimization program to construct the limiting distribution and rejection cutoffs of the GRASP procedure requires a prespecified radius.

\section{Theoretical Analysis}
\label{sec: Theory}
In this section, we give a theoretical analysis of the two-sample test. In the first part, we prove asymptotic normality of both the sample-split and cross-fit test statistics. In the second part, we derive consistent variance estimators. These two parts rigorously show the validity of Algorithms \ref{alg:sample split} and \ref{alg:crossfit} as detailed in the following theorem.  

\begin{theorem}
\label{thm:main_asymp}
    Suppose that Assumption \ref{as: ss} is satisfied. Then under the null, the sample split test statistic satisfies
    $$
    \P\left(\frac{\sqrt{n_2}}{\hat{\sigma}_{n,split}} \left[T_{n,split} - \delta-\frac{1}{2}\right] >z_{1-\alpha}\right) \leq \alpha+o(1).
$$
    Moreover, supposing Assumption \ref{as: crossfit} is satisfied, under the null the cross-fit test statistic satisfies
$$
\P\left(\frac{\sqrt{n}}{\hat{\sigma}_{tr,cross}}\left[T_{n,cross} - \delta-\frac{1}{2}\right] > z_{1-\alpha}\right) \leq \alpha + o(1).
$$
\end{theorem}

For asymptotic normality, we do not require any performance guarantees on the distinguisher procedure $\wh{D}$. However, a better distinguisher will lead to better power. In Section \ref{subsec:power}, we explain how the performance of the distinguisher affects the power. In Section \ref{sec: coupled sample}, we show that the dependence structure between the two samples does not hinder estimation of the distinguisher procedure $\wh{D}$ in some typical scenarios. 

\subsection{Asymptotic Normality}
Recall we have constructed two test statistics 
\begin{align*}
T_{n,split} = \frac{1}{n_2^2} \sum_{i, j \in \mathcal{I}_2}R_{ij}(\mathcal{D}_{n,1}), \quad
    T_{n,cross} = \frac{1}{K}\sum_{k=1}^K\frac{1}{n_k^2}\sum_{i,j \in \mathcal{I}_k}R_{ij}(\mathcal{D}_{n,-k})\,,
\end{align*}
where the quantities $R_{ij}(\cdot)$ are defined in \eqref{eq: rank ss} and \eqref{eq: rank cf} respectively and $n_k$ is the size of each fold which is assumed to be equal across all folds.

To show that both test statistics have correct asymptotic coverage, we derive the following normal approximations. 
First, we define projections of the rank-statistic, $\psi,\phi$,  by
\begin{align}
\label{eq:projections}
\begin{split}
    \phi(X,Y;\mathcal{D}) &= \E_{X',Y'}\Big[\mathds{1}\Big(\hat g_{\mathcal{D}}(X,Y)<\hat g_{\mathcal{D}}(X',Y')\Big)\Big]\\
    \psi(X',Y';\mathcal{D}) &= \E_{X,Y}\Big[\mathds{1}\Big(\hat g_{\mathcal{D}}(X,Y)<\hat g_{\mathcal{D}}(X',Y')\Big)\Big],
\end{split}
\end{align}
where $\hat g_{\mathcal{D}}$ is the distinguisher procedure $\wh D$ applied to the data $\mathcal{D}$ and the expectations are taken over $(X,Y)\sim P_X \times \Cat(\eta(X))$ and $(X',Y') \sim P_X \times \Cat(\hat \eta(X))$ respectively. Here we use the notation $\E_Z$ for expectation over $Z$ conditioning on all other randomness. Define the variances  
\begin{align}
    \label{eq: var ss}
    \sigma^2_{n,split} &= \var\Big(\phi(X,Y;\mathcal{D}_{n,1}) + \psi(X,Y';\mathcal{D}_{n,1})\Big|\mathcal{D}_{n,1}\Big)\\
    \label{eq: var crossfit}
    \sigma^2_{tr,cross} &= \var\Big(\E\left[\phi(X,Y;D_{n,-1}) + \psi(X,Y';D_{n,-1})|X,Y,Y'\right]\Big).
\end{align}
where $(X,Y) \sim P_X \times \Cat(\eta(X))$ and $Y'\sim \Cat(\hat{\eta}(X)),$ independent of the data $\mathcal{D}_n$.
In $\sigma^2_{n,split}$, $\mathcal{D}_{n,1}$ is the fitting data. As $\mathcal{D}_{n,1}$ is random, the variance term $\sigma^2_{n,split}$ is also random. In $\sigma^2_{tr,cross}$, we take an expectation over the fitting data $\mathcal{D}_n$. We note that $\sigma^2_{tr,cross}$ is deterministic. In the cross-fitting procedure, we will use notation $\sigma^2_{tr,cross}$ to emphasize that the value of $\sigma^2_{tr,cross}$ only depends on the fitting sample size, which is the same ($n-\frac{n}{K}$) across all folds. Our first result corresponds to the sample split statistic $T_{n,split}$.
\begin{assumption}
\label{as: ss}
Assume that as $n_1,n_2 \rightarrow \infty$, we have that $\sigma^3_{n,split} \sqrt{n_2} \rightarrow \infty$ in probability.
\end{assumption}

The reason for this assumption is that conditional on $\mathcal{D}_{n,1}$, $T_{n,split}$ is a two-sample U-statistic which is asymptotically normal as $n_2 \rightarrow \infty$. Assumption \ref{as: ss} is needed to show that $T_{n,split}$ is unconditionally normal as both $n_1,n_2 \rightarrow \infty$ simultaneously. This assumption is relatively mild, as $\sigma_{n,split}$ typically captures the variability from an independent triplet $(X,Y,Y')$, which should be of a constant order as long as the distingsuisher is non-degenerate. In practice we find that a $50-50$ split works well.

\begin{proposition}
    \label{thm: asymp normality ss}
    Under Assumption \ref{as: ss}, the test statistic $T_{n,split}$ given by Algorithm \ref{alg:sample split} satisfies
    $$
    \frac{\sqrt{n_2}}{\sigma_{n,split}} \Big(T_{n,split} - \mu(\mathcal{D}_{n,1})\Big) \xrightarrow{d} \N(0,1).
$$
\end{proposition}
The proof of asymptotic normality of the sample split estimator is similar to previous work such as \citet{caiAsymptoticDistributionFreeIndependence2024a}, and we defer it to Appendix \ref{sec: Proofs}. Next, we give the normal approximation for the cross-fit test statistic under the following stability assumption on $\wh{D}$.

\begin{definition}
    \label{def: subweibull} A random variable $Z$ is $(\kappa,\beta)$ sub-Weibull if for all $t\geq 0$
    $$
        \P\big(|Z|\geq t\big) \leq 2\exp\left(-\left(\frac{t}{\kappa}\right)^{1/\beta}\right).
    $$
\end{definition}

\begin{definition}
    \label{def:stability}
    Let $\mathcal{D}_n = \{Z_1,...,Z_n\}$ consist of iid samples $Z_i$ on sample space $\mathcal{Z}$ and define $\mathcal{D}_n^i = \{Z_1,.,Z_{i-1},Z'_i,Z_{i+1}..,Z_n\}$ to be $\mathcal{D}_n$ where the $i$-th sample $Z_i$ is replaced by an iid copy $Z_i'.$ Let $\ell: \bigcup_{n=1}^\infty \mathcal{Z}^n \rightarrow \R$. Define the perturb-one operator
    $$
        \nabla_i \ell(\mathcal{D}_n) = \ell(\mathcal{D}_n) - \ell(\mathcal{D}_n^i).
    $$
\end{definition}
In the following assumption, we take $Z_i$ in Definition \ref{def:stability} to be the triplet $(X_i,Y_i,Y'_i).$ 

\begin{assumption}
    \label{as: crossfit}
    Let $\wh{D}$ be the distinguisher procedure and $\hat g_n=\wh{D}(\mathcal D_n)$. Let $(X,Y) \sim P_X\times\Cat(\eta(X))$ and $(X',Y') \sim P_X\times\Cat(\hat\eta(X))$ be independent and independent of $\mathcal D_n$. 
    The following hold.
    \begin{enumerate}
        \item For any dataset $\mathcal{D}_n$, $\hat g_n(X,Y)$ and $\hat g_n(X',Y')$ have densities bounded by some $B>0$.
        \item  When $\mathcal D_n$ consists of iid samples from $P_X\times\Cat(\eta(X))\times\Cat(\hat\eta(X))$,
        both $\nabla_i \hat g_n(X,Y)$ and $\nabla_i\hat g_n(X',Y')$ are $(\log(n)^{-\beta}n^{-1/2},\beta)$ sub-Weibull for some positive constant $\beta$. Here the randomness are with respect to both that of $\hat g_n$ and the evaluating sample points $(X,Y)$, $(X',Y')$.
        \item There exists a constant $c>0$ such that $\sigma_{tr,cross}^2 > c >0.$ 
    \end{enumerate}
\end{assumption}

The first part of the assumption is a non-degeneracy condition of the distinguisher procedure. This assumption has been used in previous analysis of this type of rank-sum statistic \citep{huTwoSampleConditionalDistribution2024a}. Part 2 is a stability condition used to establish the cross-fit CLT. We require a stronger notion of sub-Weibull stability as opposed to $L_2$-stability because of the discontinuity of the ranking statistic. Our assumption is equivalent to $L_2$- stability at rate $n^{-1/2}$ up to polylog factors. For typical estimators that involve sample averages (M-estimation), under exchangeability, each sample point should contribute at most $O(1/n)$. We refer the reader to \citet[][Section 6]{leiModernTheoryCrossValidation2025} for references on stability of common algorithms.

We give some intuition for the third part of the assumption. For simplicity, consider a fixed data set $\mathcal{D}$ and the variance term $
    \var(\phi(X,Y;\mathcal{D}) + \psi(X,Y';
    \mathcal{D})\mid\mathcal{D}).
$
By the law of total variance, this is lower bounded by
$\E(\var[\phi(X,Y;\mathcal{D}) + \psi(X,Y';
    \mathcal{D})\mid X,\mathcal{D}]).
$
We can lower bound the conditional variance term by $(\phi(X,1;\mathcal{D})-\phi(X,0;\mathcal{D}))^2 \eta(X)(1-\eta(X))$
so 
\begin{equation*}
    \var\Big(\phi(X,Y;\mathcal{D}) + \psi(X,Y';
    \mathcal{D})\Big|\mathcal{D}\Big) \geq \E\Big[(\phi(X,1;\mathcal{D})-\phi(X,0;\mathcal{D}))^2 \eta(X)(1-\eta(X))\Big].
\end{equation*}

We can lower bound $\E\Big[(\phi(X,1;\mathcal{D})-\phi(X,0;\mathcal{D}))^2 \eta(X)(1-\eta(X))\Big]$ away from zero if on the subspace of $\mathcal{X}$ where $\eta(X)\neq 0,1$, $\hat g_{\mathcal{D}}$ is not degenerate, so that $(\phi(X,1;\mathcal{D})-\phi(X,0;\mathcal{D}))^2 > 0$. In general, degeneracy of $\hat g_{\mathcal{D}}$ is not an issue since if $\hat g_{\mathcal{D}}$ is degenerate, random tie-breaking as in \eqref{eq: random tie-breaking}, guarantees a lower bound for the variance.

\begin{proposition}
    \label{thm: crossfit asymp norm} Under Assumption \ref{as: crossfit}, the statistic $T_{n,cross}$ given by Algorithm \ref{alg:crossfit} satisfies
    $$
        \frac{\sqrt{n}}{\sigma_{tr,cross}} \left(T_{n,cross} - \frac{1}{K}\sum_{k=1}^K\mu(\mathcal{D}_{n,-k})\right) \xrightarrow{d} \N(0,1).
    $$
\end{proposition}

 In the remainder of this subsection, we sketch the proof of asymptotic normality of the cross-fit statistic. The main technical tool is the following central limit theorem for cross-validation with random centering which is given in \citet{leiModernTheoryCrossValidation2025}. 
\begin{theorem}
    \label{thm: stability clt} 
    Let $Z_1,...,Z_n$ be iid random variables supported on some space $\mathcal{Z}$. Let $\mathcal{D}_n:=\{Z_1,...,Z_n\}$ and $Z$ be another iid copy. We split the data into $K$ folds with indices $\mathcal{I}_1,...,\mathcal{I}_K$ each of size $\frac{n}{K}$ and use $\mathcal{D}_{n,-k}$ to denote the data outside of the $k$-th fold. Let $\mathcal{F}$ be a parameter space and let
    $
        \hat{f}:\bigcup_{n=1}^\infty \mathcal{Z}^{n} \rightarrow \mathcal{F}
    $
    be a symmetric estimation procedure. Then for any function $\ell:\mathcal{Z}\times \mathcal{F} \rightarrow \R^{\geq0}$ satisfying the stability condition
    $$
        \norm{\nabla_i \ell(Z,\hat{f}(\mathcal{D}_{n}))}_2 = o\left(\frac{\sigma_n}{\sqrt{n}}\right),
    $$ where
    $
        \sigma_n^2 = \var\left(\E[\ell(Z,\hat{f}(\mathcal{D}_n))|Z]\right) < \infty,
    $ we have, with $n_{tr}:= n(1-1/K)$,
    $$
        \frac{\sqrt{n}}{\sigma_{n_{tr}}}\left(\frac{1}{n}\sum_{k=1}^K\sum_{i\in\mathcal{I}_{n,k}}\ell(Z_i,\hat{f}(\mathcal{D}_{n,-k})) - \frac{1}{K}\sum_{k=1}^K\E_Z \ell(Z,\hat{f}(\mathcal{D}_{n,-k}))\right) \xrightarrow{d} \N(0,1),
    $$
    where $\E_Z$ denotes expectation with respect to $Z$ conditioning on all other randomness.
\end{theorem}

There are two major steps we need to be able to apply Theorem \ref{thm: stability clt} to the test statistic $T_{n,cross}$. First, notice that Theorem \ref{thm: stability clt} is stated for iid samples, while the test statistic $T_{n,cross}$ is in the form of a two-sample U-Statistic. We use the following Hoeffding-like decomposition to reduce the U-statistic to the iid setting.
Recall the projections $\phi,\psi$ defined in (\ref{eq:projections}). On each fold $k=1,...,K$, we show that 
\begin{equation*}
    \frac{1}{n_k^2} \sum_{i,j\in\mathcal{I}_{n,k}}R_{ij}(\mathcal{D}_{n,-k}) \approx \frac{1}{n_k}\sum_{i\in\mathcal{I}_{n,k}} \Big(\phi(X_i,Y_i,\mathcal{D}_{n,-k}) + \psi(X_i,Y'_i,\mathcal{D}_{n,-k})\Big),
\end{equation*}
where the right side can be viewed as an average of functions of the the iid triplet $(X_i,Y_i,Y'_i)$. When the two samples are conditionally independent, this decomposition is exactly the Hoeffding decomposition \citep{vaartAsymptoticStatistics1998a}. We verify that the decomposition still holds for our dependency structure. Under the above decomposition, to apply Theorem \ref{thm: stability clt} with the correspondence $Z=(X,Y,Y')$, we want to set the $\ell$ function to be
$$
\ell((x,y,y'),\mathcal D) =   \phi(x,y;\mathcal{D}) + \psi(x,y';\mathcal{D}).
$$
To show the required stability, we show that the stability of the distinguisher procedure given in Assumption \ref{as: crossfit} implies $\norm{\nabla_i [\phi(X,Y;\mathcal{D}_n) + \psi(X,Y';\mathcal{D}_n)]} = o(n^{-1/2})$. With the conditions of Theorem \ref{thm: stability clt} satisfied, the proof of Proposition \ref{thm: crossfit asymp norm} can be reduced to an application of the cross validation CLT. The full proof is given is Appendix \ref{sec: Proofs}.



\subsection{Variance Estimation}
To apply the asymptotic normality results in the previous section, we need to provide consistent estimators of the variance terms. Let's first consider the sample split case
\begin{align*}
\sigma^2_{n,split} &= \var\left(\phi(X,Y;\mathcal{D}_{n,1}) + \psi(X,Y';\mathcal{D}_{n,1})\right).
\end{align*}
Naturally, we suggest using the empirical variance.
The issue is that we need to empirically estimate the U-statistic projections $\phi,\psi$ from data.
Recall the definition of the projections (\ref{eq:projections}).
A natural estimator for the projections $\phi(X,Y;\mathcal{D}_{n,1})$ and $\psi(X,Y';\mathcal{D}_{n,1})$ is thus
\begin{align}
\begin{split}
\label{eq:emp_proj}
    \wh{\phi}(X,Y;\mathcal{D}_{n,1}) &= \frac{1}{n_2} \sum_{i \in \mathcal{I}_{n,2}} \mathds{1}\Big(\hat g_{n,1}(X,Y)<\hat g_{n,1}(X_i,Y'_i)\Big)\\
    \wh{\psi}(X',Y';\mathcal{D}_{n,1}) &=\frac{1}{n_2} \sum_{i \in \mathcal{I}_{n,2}}\mathds{1}\Big(\hat g_{n,1}(X_i,Y_i)<\hat g_{n,1}(X',Y')\Big).
\end{split}
\end{align}
We can use the empirical version of $\sigma_{n,split}^2$ as an estimator of the variance 
\begin{equation}
    \label{eq: var est split}
    \wh{\sigma}^2_{n,split} =  \frac{1}{n_2-1}\sum_{i\in \mathcal{I}_{2,n}} \Big(\wh{\phi}(X_i,Y_i;\mathcal{D}_{n,1}) + \wh{\psi}(X_i,Y_i';\mathcal{D}_{n,1}) - 2T_{n,split}\Big)^2.
\end{equation}

\begin{proposition}
    \label{prop: var est split} Under Assumption \ref{as: ss}, the variance estimator \eqref{eq: var est split} satisfies 
    $
        \frac{\wh{\sigma}_{n,split}^2}{\sigma^2_{n,split}} \xrightarrow{p} 1.
    $
\end{proposition}

Next, we move on to estimating the variance $\sigma^2_{tr,cross}$. The key idea is to leverage Efron-Stein inequality which says that the contribution to the variance by the evaluation variables $X,Y,Y'$ dominates the contribution to the variance by the fitting data $\mathcal{D}_{n,-k}$ under our stability assumption. This implies that we can aggregate the empirical variance across the $K$ folds. For each fold $k=1,...,K$, we can estimate the variance 
$$
    \wh{\sigma}^2_{n,-k} = \frac{1}{n_k-1}\sum_{i\in \mathcal{I}_{n,k}} \left(\wh{\phi}(X_i,Y_i;\mathcal{D}_{n,-k})+\wh{\psi}(X_i,Y'_i;\mathcal{D}_{n,-k}) - 2T_{n,-k}\right)^2,
$$
where 
$
    T_{n,-k} = \frac{1}{n_k^2}\sum_{i,j \in \mathcal{I}_k}R_{ij}(\mathcal{D}_{n,-k})
$ and $\wh{\phi}(X,Y;\mathcal{D}_{n,-k})$, $\wh{\psi}(X',Y';\mathcal{D}_{n,-k})$ are the corresponding cross-fit versions of \eqref{eq:emp_proj}. We can define the final estimate by aggregation
\begin{equation}
    \label{eq: var est cross}
    \wh{\sigma}^2_{tr,cross} = \frac{1}{K} \sum_{k=1}^K \wh{\sigma}^2_{n,-k}\,.
\end{equation}

\begin{proposition}
    \label{prop: var est cross}
    Under Assumption \ref{as: crossfit}, the crossfit variance estimator satisfies 
    $
        \frac{\wh{\sigma}_{tr,cross}^2}{\sigma^2_{tr,cross}} \xrightarrow{p} 1.
    $
\end{proposition}

The CLTs given in Propositions \ref{thm: asymp normality ss}, \ref{thm: crossfit asymp norm} with the consistent variance estimates given in Propositions \ref{prop: var est split}, \ref{prop: var est cross} imply Theorem \ref{thm:main_asymp}, prove the asymptotic validity of both testing procedures.

\subsection{Power}\label{subsec:power}
So far we have established type-1 validity under mild assumptions on either the sample sizes in each split for the sample split procedure or stability for the cross-fit procedure to ensure. In particular, the quality of the distinguisher does not play a role in the validity of the test. In this subsection, we show that the quality of the distinguisher is important for the power. At a high level, the main source of conservativeness in our test is that we target a stochastic lower bound of $\AUC(L)$ by $\AUC(\hat g)$ where $\AUC(L)$ is the area under the ROC curve of the true likelihood ratio $L$. The following result, which follows from \citet{caiAsymptoticDistributionFreeIndependence2024a} [Supplement, Lemma 3] and (\ref{eq: AUC}), characterizes this conservativeness.

\begin{proposition}
\label{prop:auc_cons}
    For any distinguisher $\wh{g},$ with likelihood ratio transformation $\wh{L}=\wh{g}/(1-\wh{g})$, we have
    \begin{equation*}
        \AUC(L) - \AUC(\wh{L}) \leq 2\norm{\wh{L}-L}_{L_1}.
    \end{equation*}
\end{proposition}

Here $\|\cdot\|_{L_1}$ denotes the $L_1$-norm under the distribution induced by the random pair $(X,Y)$. The above gives support of using the cross-fitting procedure over sample-splitting. As cross-fitting effectively allows for a larger sample size to train $\hat g$, we expect the quality of $\hat g$ to be better when cross-fitting leading to a less conservative test.

Consider the case where we want to evaluate a fixed classifier $\hat{\eta}$ with true separation $\rho(P_\eta,P_{\hat{\eta}}) = \delta^*,$ where we use $P_\eta,P_{\hat{\eta}}$ as shorthand for the distributions corresponding to the true $\eta$ and the classifier $\hat{\eta}$. The key question we want to answer is: 
\begin{quote}
What is the largest $\delta \in [0,\delta^*]$ such that we have nontrivial power against the null
$
    H_0 : \rho(P_\eta,P_{\hat{\eta}} ) < \delta?
$
\end{quote}

Using Proposition \ref{prop:auc_cons}, the following corollary shows how to answer this question given the quality of the distinguisher. For brevity, we focus on the cross-fit procedure and note that an analogous result can be stated for the sample-split procedure. For notation we use $K$ folds with each fold having equal size $n_k=n/K$. For the given distinguisher $\wh D$, we use $\hat g_{n,-k}$ to denote the distinguisher applied to the data outside the $k$-th fold. Let $\wh L_{n,-k}=\hat g_{n,-k}/(1-\hat g_{n,-k})$ be the corresponding likelihood ratio estimate.


\begin{corollary}
    \label{cor:power}
    Consider the $K$-fold cross-fit test statistic with a constant $K$. Suppose Assumption \ref{as: crossfit} holds and that $\P\Big(\norm{\wh{L}_{n,-k}-L}_{L_1} > r_n \Big) \rightarrow 0$ for some sequence of rates $r_n \geq 0$. Then choosing tolerance parameter $\delta_n$ such that $\delta^*-\delta_n - 2r_{n} = \omega(1/\sqrt{n})$,  the cross-fit test statistic diverges in probability
    \begin{equation*}
        \frac{\sqrt{n}}{\hat{\sigma}_{tr,cross}}\Big(T_{n,cross} - \delta_n - \frac{1}{2}\Big)\rightarrow \infty.
    \end{equation*}
    As a consequence, the rejection rate converges to $1$.
\end{corollary}

As a further consequence, note that when the distinguisher procedure is consistent ($r_n \rightarrow 0$), asymptotically we can take $\delta_n$ to approach the optimal value $\delta^\ast.$ In the case where we cannot garentee consistent distinguisher estimation, i.e $r_n \rightarrow r \neq 0$, asymptotically the maximum tolerance $\delta_n$ we can take will be off by a bias no larger than $2r$.

\subsection{Classification With Coupled Samples}
\label{sec: coupled sample}
The power of our test is crucially dependent on the estimated distinguisher $\hat{g}$, which is typically constructed using an off-the-shelf probabilistic classification algorithm $\hat D$. The theoretical guarantees for these classification methods are usually stated in the case of two independent samples. In our setting, the two samples $\{(X_i,Y_i)\}$ and $\{(X'_j,Y'_j)\}$ are dependent since $X_i = X'_i$. Are common off-the-shelf methods effective in this dependent data setting? \citet{caiAsymptoticDistributionFreeIndependence2024a} showed that M-estimation methods, in both low and high dimensional settings, are effective in a dependent data setting, where the second sample is derived from a cyclic permutation of the first. In our present setting, the dependency structure is different. The intuition that off-the-shelf classifiers are effective in our setting is that since we know that the $X$-marginal has no effect on the class label $C$, having the same values in the $X$'s among the two samples should not drastically affect the performance. In this subsection, we provide theoretical evidence through a common high-dimensional sparse regression setting. The takeaway is whenever off-the-shelf classifiers are effective in the i.i.d. case, they are also effective in our dependency structure. 

Suppose we wish to estimate $\hat g$ through the optimization problem
\begin{equation}
    \label{eq: m-est global}
    \min_{g} \E\bigg[(1-g(X',Y'))^2+(0-g(X,Y))^2\bigg].
\end{equation}
Let $e_1,...,e_{K_n}:\R^d\times\mathcal{L}\rightarrow\R$ be a basis expansion where $K_n$ is the size of the expansion which is allowed to depend on the sample size $n$. To simplify the discussion, we assume that the optimizer $g^*$ has a sparse representation with respect to this basis. Similar results can be derived for other types of sparsity, such as $L_1$ sparsity, using the same strategy.  Without loss of generality, we consider the case of $K_n\rightarrow\infty$, because if $K_n$ is bounded, we can always add additional basis functions with zero coefficients.

\begin{assumption}
    \label{as: sparse basis rep} Assume that $g^*$ has a sparse basis representation
    $$
        g^*(x,y) = \sum_{i=1}^{K_n} \beta_i e_i(x,y),
    $$
    where $K_n$ denotes the number of basis elements used and $\beta \in \R^{K_n}$ is $s$-sparse $(\norm{\beta}_0=s)$.
\end{assumption}

Now suppose that we have data $\{(X_i,Y_i)\}_{i=1}^n$ and $\{(X'_i,Y'_i)\}_{i=1}^n$ such that $X_i=X'_i$. We can express the empirical form of (\ref{eq: m-est global}) as a linear regression problem with a design matrix $\Xi$ defined by
$$
\Xi
=
\begin{pmatrix}
\bigl(\mathfrak e(X_i,Y_i)\bigr)_{i=1}^n\\[2pt]
\bigl(\mathfrak e(X'_i,Y'_i)\bigr)_{i=1}^n
\end{pmatrix}
\in \mathbb R^{2n\times K_n}, \quad \mathfrak{e}(X_j,Y_j) = (e_i(X_j,Y_j))_{i\in[K_n]} \in \R^{K_n}
$$ 
and response  $Z := (\underbrace{0,...,0,}_n\underbrace{1,...,1}_n)$. The regularized empirical form of (\ref{eq: m-est global}) can be written as 
\begin{equation}
\label{eq: lasso problem}
    \argmin_{\beta} \frac{1}{4n} \norm{Z-\Xi\beta}_2^2 + \lambda\norm{\beta}_1,
\end{equation}
where $\lambda>0$ is a regularization parameter. We make the following restricted eigenvalue condition which is a standard assumption in LASSO analysis.
\begin{assumption}
    \label{as: LASSO}
    Let $S$ be the support of $\beta$. For a constant $\alpha\geq 1$, define the set
    $$
        \C_\alpha(S) = \{\Delta \in \R^{K_n}:\norm{\Delta_{S^c}}_1\leq \alpha\norm{\Delta_{S}}_1\}.
    $$
    We say that the design matrix $\Xi$ satisfies the $(\kappa,\alpha)$ restricted eigenvalue (RE) condition if
    \begin{equation}
        \label{eq:re}
        \frac{1}{n}\norm{\Xi\Delta}_2^2 \geq \kappa\norm{\Delta}_2^2
    \end{equation}
    for all $\Delta \in \C_\alpha(S)$.
    We assume the design matrix $\Xi$ satisfies the restricted eigenvalue condition with parameters $(\kappa,3)$.
\end{assumption}

In general, it is difficult to verify the restricted eigenvalue condition. There exist results which show that in the random design setting, it is possible to show the restricted eigenvalue condition holds with high probability \citep[][Theorem 7.16]{wainwrightHighDimensionalStatisticsNonAsymptotic2019}. These results require the rows of the design matrix to be iid. Our setting differs in two ways. The first $n$ rows and the last $n$ rows come from two different distributions and they are dependent with eachother. We give a discussion on why these results can be extended to our sample setting. For notation, define $\Xi_1,\Xi_2 \in \R^{n\times K_n}$ by
\begin{equation*}
    \Xi_1 = \begin{pmatrix}
        \mathfrak{e}(X_1,Y_1)\\
        \cdots\\
        \mathfrak{e}(X_n,Y_n)
    \end{pmatrix} \qquad \Xi_2 = \begin{pmatrix}
        \mathfrak{e}(X'_1,Y'_1)\\
        \cdots\\
        \mathfrak{e}(X'_n,Y'_n)
    \end{pmatrix}.
\end{equation*}
Under this notation, the RE condition (\ref{eq:re}) is 
\begin{equation}
\label{eq:re sep}
    \frac{1}{n} \norm{\Xi_1 \Delta}_2^2 + \frac{1}{n} \norm{\Xi_2 \Delta}_2^2 \geq \kappa \norm{\Delta}_2^2,
\end{equation}
for all $\Delta \in \C_\alpha(S)$.
Now each $\Xi_1,\Xi_2$ are design matrices where each row are iid. If these design matrices satisfy RE condition with slightly different parameters $(\kappa/2,3)$ with high probability, we can use union bound to see that the combined design matrix $\Xi$ satisfies RE condition with parameters $(\kappa,3)$ with high probability. We can derive the following consistency result in the setting where the basis expansion functions are bounded. In particular, when $\lambda \asymp \sqrt{(\log K_n)/n}$, $\hat{\beta}$ is consistent for $\beta^*$ as long as $\frac{\log(K_n)}{n} \rightarrow 0$.

\begin{theorem}
    \label{thm:consistency}
        Suppose that each basis function $e_i,i=1,...,K_n$ is bounded by $[-B,B]$ and $\lambda \asymp \sqrt{(\log K_n)/n}$. With probability at least $1-{K_n}^{-1}$,
    $
    \norm{\hat{\beta}-\beta^*}_2 \leq C\sqrt{(s\log(K_n))/n}.
    $
\end{theorem}

\section{Numerical Simulations}
\label{sec: Sim}
In this section, we consider two simulation settings to evaluate the performance of our proposed test procedures. The first is a logistic regression setting similar to \citet{javanmardGRASPGoodnessoffitTest2024}. The second setting is a sparse logistic setting where we illustrate the adaptivity of our method to underlying structure.

\subsection{Logistic Regression Setting}

We first evaluate the empirical performance of our Goodness-of-Fit test in a similar logistic regression setting as in \citet{javanmardGRASPGoodnessoffitTest2024}. The data is generated as follows:
\begin{enumerate}
    \item Generate a single $\theta^* \sim \mathcal{N}\left(0,(0.25)^2\cdot\I_{200}\right)$ to be used in all experiments.
    \item Set the true conditional probability function and the estimated conditional probability as
$$
    \eta(x) = \frac{1}{1+\exp(-x^\top\theta^*)} \qquad  \hat{\eta}(x) = \frac{1}{1+\exp(-x^\top\hat{\theta})},
$$
where $\hat{\theta} = \theta^*$ in the null case and $\hat{\theta} = -\theta^*$ in the alternative.
    \item Generate a holdout set $\{(X_i,Y_i)\}_{i=1}^n$ where
$
    X_1,...,X_n \overset{iid}{\sim}  \N(0,\I_{200}),Y_i \sim \text{Bern}(\eta(X_i)),
$
where $n=1000, 2000, 3000$.
\end{enumerate}

For a tolerance parameter $\delta>0$, we are interested in testing the hypothesis
$
    H_0: \rho(P_0,P_1) < \delta,
$
where $P_0 := \N\left(0,\I_{200}\right) \times \text{Bern}(\eta(X))$ and $P_1 := \N\left(0,\I_{200}\right) \times \text{Bern}(\hat{\eta}(X))$. In the following simulations, we compare three test procedures: sample-split, cross-fit, and asymptotic GRASP \citep{javanmardGRASPGoodnessoffitTest2024}. Across all simulations, we use a 50-50 split for the sample-split method, and we use 5 folds for the cross-fit method. 

We use the following method to construct the distinguisher. Suppose that we have the samples $(X_i,Y_i,C_i)$ for $i=1,...,n$ where $C_i=0,1$ denotes the distribution from which the sample comes from. Let $\mathcal{D}_0 := \{(X_i,C_i):Y_i = 0\}$ and $\mathcal{D}_1 := \{(X_i,C_i):Y_i = 1\}$. Using data $\mathcal{D}_0$, train a logistic classifier, which we call $\hat g_0$ and using data $\mathcal{D}_1$, train a  logistic classifier, which we call $\hat g_1$. We can combine both $\hat g_0$ and $\hat g_1$ into a single distinguisher by the rule $\hat g(x,y) = \hat g_y(x).$ The reason for splitting the distinguisher based on the value $Y$ is that we know that the marginal $P_X$ alone cannot distinguish between the two samples.


\paragraph*{Type-1}
For simulations under null, we consider the exact null
$
    \rho(P_0,P_1) = 0,
$
which corresponds to the setting $\hat{\theta} = \theta^*$. We set the significance level at $\alpha=0.05$ and evaluate the empirical type-1 error over 500 iterations with sample size $n=1000,2000,3000$. In Table \ref{tb: null}, we show the average number of rejections in the 500 trials for each of the three compared methods at the various sample sizes with $\alpha=0.05$. We find that across all settings, all three methods control type-1 error close to the nominal $\alpha=0.05$ level. The correct type-1 control of the cross-fitting statistic suggests that the stability assumption (\ref{as: crossfit}) is satisfied.

\begin{table}[H]
\centering
\begin{tabular}{l|lll}
n=(1000,2000,3000) & Sample split          & Cross-fit             & GRASP                 \\ \hline
Type-1 error       & (0.054, 0.054, 0.044) & (0.048, 0.052, 0.036) & (0.042, 0.040, 0.066)
\end{tabular}
\caption{The empirical type 1 error and average test statistic of 500 simulations under the null. All three methods reasonably controls type-1 error at the given $\alpha=0.05$ level.}
\label{tb: null}
\end{table}


\paragraph*{Power} In the alternative setting we set $\hat{\theta} = -\theta^*$. Numerically, we approximate the true separation $\delta^*$ to be
$
\delta^* :=\rho(P_0,P_1) \approx 0.45.
$ We will use the tolerant null
$
    H_0: \rho(P_0,P_1) < \delta.
$
We consider sample sizes of $n=1000,2000,3000$ at level $\alpha=0.05$. Following the setup in \cite{javanmardGRASPGoodnessoffitTest2024}, we record the empirical power over 50 trials, which is plotted in Figure \ref{fig: power} (left). The x-axis of the plot is the ratio $\frac{\delta}{\delta^*}$ between the chosen tolerance $\delta$ and the true separation $\delta^* \approx 0.45$. Choosing $\delta=0$ corresponds to a tolerance ratio of $0$ and setting $\delta=\delta^*$ corresponds to a tolerance ratio of $1$.

In the left plot in Figure \ref{fig: power}, we also plot empirical results for GRASP, where we test the tolerant null
$
    H_0: \rho_{\rm tv}(P_0,P_1) < \delta,
$
using the same experimental setting. The true separation under the total variation distance is
$
    \delta^*_{\rm tv} \approx 0.69,
$
and we again plot the power against the tolerance ratio $\frac{\delta}{\delta^*_{\rm tv}}$. There is a slight discrepancy in comparing our simulations results as GRASP uses a $f$-divergence such as the distance from the TV as the measure of separation while we use AUC. However, by Proposition \ref{prop: tv and auc}, we see that our distance is nearly proportional to the TV distance. Thus, we compare the empirical power at various tolerance ratios. 
\begin{figure}[]
\centering
\includegraphics[width=8cm]{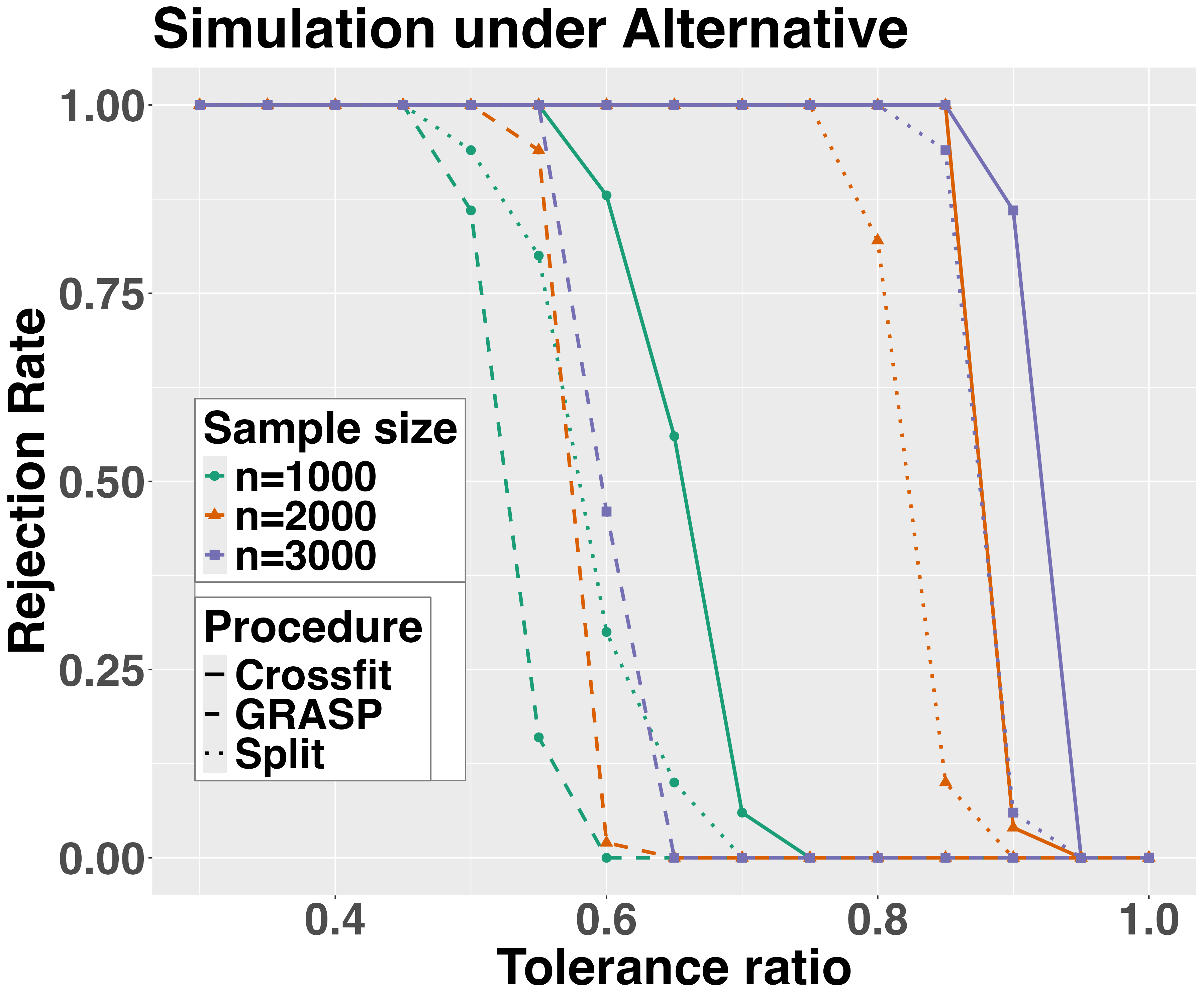}\includegraphics[width=8cm]{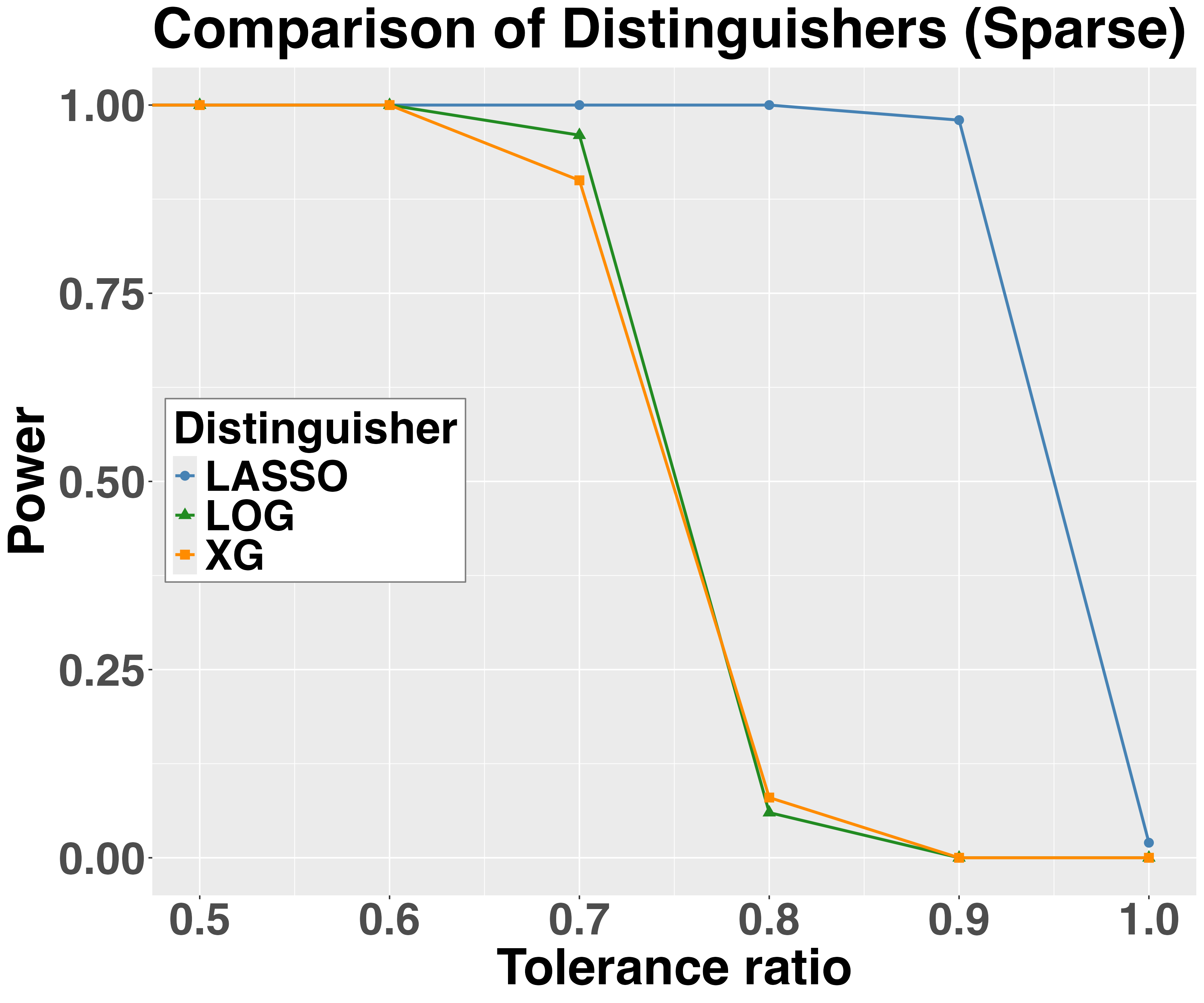}
\caption{Left: Comparison of sample split method, cross-fit method and GRASP. The empirical power of the three methods indicated by line type are plotted across multiple sample sizes. Right: Comparison of empirical power of three distinguisher choices in sparse setting. The adaptive LASSO distinguisher performs the best.}
\label{fig: power}
\end{figure}

The results in Figure \ref{fig: power} show that across all sample sizes, the cross-fit test statistic is more powerful compared to the sample split test statistic. This is expected as the power of our method depends heavily on the quality of the estimated distinguisher $\hat g$. Cross-fitting allows us to use a substantially larger sample size to train $\hat g$ which should lead to improvements in power. Overall, both the cross-fit and sample split methods perform better than GRASP.

\begin{remark}
    In our comparison with the GRASP method, we use a model agnostic choice of scoring function for GRASP as in the corresponding simulations in \citet[][Example 5.1/5.2]{javanmardGRASPGoodnessoffitTest2024}.  \citet[][Section 4]{javanmardGRASPGoodnessoffitTest2024} discusses using model based scoring functions in the model-X setting with auxiliary unsupervised data. As we do not assume access to a model-X setting, we do not make comparisons to these types of scoring functions. 
\end{remark}

\subsection{Sparse Setting}
In our second setting, we illustrate the adaptivity of our proposed testing procedure by evaluating the power of the cross-fit procedure in a sparse logistic setting. We focus on the cross-fit procedure in this section as it is more powerful than the sample split procedure. The data generation process is the same as in Setting 1 except that the vector $\theta^* = (1,1,1,1,1,0,...,0) \in \R^{200}$ is now sparse. To adapt to the sparsity we use a logistic LASSO distinguisher using the same procedure as in the logistic regression setting.


To understand the effect of choice of distinguisher, we compare the LASSO distinguisher with two other distinguisher procedures. First, we compare it to a logistic regression distinguisher, which is the correct parametric model but does not account for the sparsity. We further compare it to an XgBoost distinguisher, a common off-the-shelf method which also does not account for the sparsity. The results for the sample size of $n=1000$ are plotted on the right in Figure \ref{fig: power}. As in the first setting, we use 5 folds for cross-fitting and consider $\alpha=0.05$. In general, we find that when the distinguisher is adaptive to the underlying sparsity, we see a significant improvement in power reflecting the idea that our test is adaptive to additional distributional information.

\section{Real Data Examples}
\label{sec: Data}
Many machine learning datasets are used to benchmark different learning methods. Normally, the classification accuracy of different machine learning methods is used for comparison. In this section, we evaluate different classification methods using the goodness-of-fit testing approach on two popular benchmarking datasets, MNIST and Fashion MNIST, and compare the qualitative takeaways compared to solely using classification accuracy.  

Both datasets are used as benchmarks for multi-class classification. In both datasets, we are given a $28 \times 28$ grayscale image, which we can model as a feature vector $X\in\R^{784}$, where each entry denotes the grayscale level of one of the $784$ pixels. Both datasets also come with 10 total labels which we denote by $\ell\in \mathcal{L}=\{0,...,9\}$. For the MNIST dataset, the labels correspond to handwritten digits. For the Fashion MNIST dataset, the labels correspond to different items of clothing. 







We can model the data-generating distribution of both datasets by
\begin{equation*}
    (X,\ell) \sim P_X \times \Cat(\eta(X)),
\end{equation*}
where $P_X$ is a the marginal probability distribution of the grayscale level of the pixels, supported on $\R^{784}$ and $\eta: \R^{784} \rightarrow \Delta_9$ is a function from the features to the probability simplex on the set of labels $\{0,...,9\}$. 

The classification task is to construct an estimate $\hat{\eta}$ of the function $\eta$. We consider evaluating multiclass classifiers $\hat{\eta}_{\rm lin}$, $\hat{\eta}_{\rm rf}$ and $\hat{\eta}_{\rm xg}$, which correspond to logistic regression, random forest and XgBoost, respectively. The goal is to test the hypothesis
\begin{equation*}
    H_0: \rho\Big(P_X \times \Cat\big(\eta(X)\big), P_X \times \Cat\big(\hat{\eta}_{\mathsf{a}}(X)\big)\Big) < \delta,
\end{equation*}
where $\mathsf{a} \in \{\rm lin,\rm rf, \rm xg\}$. For the purpose of this experiment, we use 10000 samples as a training set and the remaining 60000 samples for evaluation. We focus on the cross-fit procedure with $5$ folds as we know that the cross-fit procedure is more powerful. The distinguisher procedure uses Xgboost.

 Table \ref{tb:mnist} records the classification accuracy and the GoF test results for each of the three methods. In the following results, we report the smallest radius $\delta$ which the GoF test does not reject at nominal type-1 level $0.05$. A smaller radius implies that the classifier is closer to the data-generating distribution. By looking at the classification accuracy alone, we conclude that the nonlinear methods perform approximately the same and both seem significantly better than the linear method. However, the GoF results reveal that among the two nonlinear methods, Xgboost fits the data generating distribution better than random forests, which suggests that XgBoost is the better classifier for this task. 

\begin{table}[h]
\centering
\begin{tabular}{l|ll}
$(\hat{\eta}_{\rm lin},\hat{\eta}_{\rm rf},\hat{\eta}_{\rm xg})$ & Accuracy           & GoF Rejection Radius \\ \hline
MNIST                                                            & (0.84, 0.95, 0.95) & (0.15, 0.12, 0.014)  \\
FMNIST                                                           & (0.73, 0.86, 0.87) & (0.16, 0.070, 0.021)
\end{tabular}
\caption{Comparison of classification accuracy and GoF rejection radius}
\label{tb:mnist}
\end{table}

\section{Discussion}
In this work, we evaluate classifiers using goodness-of-fit tests as an alternative to classification accuracy. Leveraging ideas from two-sample tests by classification, outcome indistinguishability, and cross-validation central limit theorems, we propose a test procedure which can assess black box classifiers under few assumptions on the classification procedure, which is effective in high-dimensional settings, is easy and efficient to implement, and can handle multi-class settings. We give some possible further research directions. In our setting, we consider assessing the performance of a given black-box estimate $\hat{\eta}$ beyond the use of classification accuracy. Another setting in which classification accuracy is used is in parameter tuning. In this setting, cross-validation may be used in conjunction with classification accuracy. It would be valuable to understand how to implement cross-validation and the GoF approach together to tune parameters in general classification models. This work falls into the general setting of inference of degenerate U-statistics using sample splitting. In general, sample splitting is a known technique to use for such problems \citep{kimDimensionagnosticInferenceUsing2024}. It would be interesting to see if cross-validated central limit theorems \citep{leiModernTheoryCrossValidation2025} can be used to allow cross-fitting to be used instead of sample splitting for better efficiency. 


\section*{Code Availability}
Code to implement the simulations can be found at \url{https://github.com/yuch44/GoF-classifier}. All data used are publically available.

\appendix
\section{Proofs}

\subsection{Cross-fit Asymptotics}
In this section we give the proofs to Lemmas \ref{lm: hajek decomposition} and \ref{lem: stability u statistic} used in proving Theorem \ref{thm: crossfit asymp norm}.

\begin{lemma}
    \label{lm: hajek decomposition} For any fold $k =1,...,K$ let
    \begin{align}
    \begin{split}
        \Phi_k &= \frac{1}{n_k}\sum_{i \in \mathcal{I}_{n,k}}\phi(X_i,Y_i,\mathcal{D}_{n,-k}) - \mu(\mathcal{D}_{n,-k})\\
        \Psi_k &= \frac{1}{n_k}\sum_{i \in \mathcal{I}_{n,k}}\psi(X'_i,Y'_i,\mathcal{D}_{n,-k}) - \mu(\mathcal{D}_{n,-k}).
    \end{split}
    \end{align}
    Then,
    \begin{equation*}
        \frac{\sqrt{n}}{\sigma_{tr,cross}}\left(\frac{1}{n_k^2} \sum_{i,j\in\mathcal{I}_{n,k}}R_{ij}(\mathcal{D}_{n,-k}) -\mu(\mathcal{D}_{n,-k}) - \Phi_k-\Psi_k\right) \xrightarrow{p} 0.
    \end{equation*}
\end{lemma}

\begin{proof}
For brevity of notation, we will use $\mu_{n,k} := \mu(\mathcal{D}_{n,-k})$. To simplify notation we set $\sigma_n := \sigma_{tr,cross}$.

By definition, we want to show that 
\begin{align*}
&\frac{\sqrt{n}}{\sigma_n n_k^2}\sum_{i,j\in\mathcal{I}_{n,k}}[R_{ij}(\mathcal{D}_{n,-k}) -\mu_{n,k}] -\frac{\sqrt{n}}{\sigma_n n_k}\sum_{i\in \mathcal{I}_{n,k}} [\phi(X_i,Y_i;\mathcal{D}_{n,-k})-\mu_{n,k}]-\frac{\sqrt{n}}{\sigma_n n_k}\sum_{i\in \mathcal{I}_{n,k}}[\psi(X'_i,Y'_i;\mathcal{D}_{n,-k})-\mu_{n,k}]\\
&\xrightarrow{p} 0.
\end{align*}
We can rewrite the left side as
\begin{align*}
&\underbrace{\frac{\sqrt{n}}{\sigma_n n_k^2}\sum_{i,j\in\mathcal{I}_{n,k}}[(R_{ij}(\mathcal{D}_{n,-k}) -\mu_{n,k}) -(\phi(X_i,Y_i;\mathcal{D}_{n,-k})-\mu_{n,k}) - (\psi(X'_j,Y'_j;\mathcal{D}_{n,-k})-\mu_{n,k})]}_{(*)}.
\end{align*}
We show that $(*) \xrightarrow{p} 0.$
        We need to show that for all $\epsilon >0$,
    \begin{align*}
    &\P\left(\frac{\sqrt{n}}{\sigma_n n_k^2}\sum_{i,j\in\mathcal{I}_{n,k}}[(R_{ij}(\mathcal{D}_{n,-k}) -\mu_{n,-k}) -(\phi(X_i,Y_i;\mathcal{D}_{n,-k})-\mu_{n,k}) - (\psi(X'_j,Y'_j;\mathcal{D}_{n,-k})-\mu_{n,k})]>\epsilon\right)\\
    &\rightarrow 0.
    \end{align*}
    
    By conditioning on the out of fold data $\mathcal{D}_{n,-k}$, we see that the left side is equal to
    \begin{align*}
    &\E\Bigg[\P\Bigg(\frac{\sqrt{n}}{\sigma_n n_k^2}\sum_{i,j\in\mathcal{I}_{n,k}}\Big[(R_{ij}(\mathcal{D}_{n,-k}) -\mu_{n,k})\\
    &\quad-(\phi(X_i,Y_i;\mathcal{D}_{n,-k})-\mu_{n,k}) - (\psi(X'_j,Y'_j;\mathcal{D}_{n,k})-\mu_{n,k})\Big]>\epsilon\Bigg\rvert\mathcal{D}_{n,-k}\Bigg) \Bigg].
    \end{align*}
    The inner probability can now be bounded using Chebyshev's inequality.
    We have
    \begin{align*}
        &\E\left[\left(\frac{\sqrt{n}}{\sigma_n n_k^2}\sum_{i,j\in\mathcal{I}_{n,k}}\underbrace{(R_{ij}(\mathcal{D}_{n,-k}) -\mu_{n,k}) -(\phi(X_i,Y_i;\mathcal{D}_{n,-k})-\mu_{n,k}) - (\psi(X'_j,Y'_j;\mathcal{D}_{n,-k})-\mu_{n,k})}_{=:\xi_{i,j}}\right)^2\Bigg\rvert\mathcal{D}_{n,-k}\right]\\
        &=\frac{n}{\sigma_n^2 n_k^4} \sum_{(i,j)(i',j')} \E(\xi_{i,j} \xi_{i',j'}|\mathcal{D}_{n,-k}).
    \end{align*}
We need to count the number of $(i,j),(i',j')$ pairs where $ \E(\xi_{i,j} \xi_{i',j'}|\mathcal{D}_{n,-k}) \neq 0.$ Conditioning on $\mathcal{D}_{n,-k}$, the expectation is over the variables $(X_i,Y_i),(X'_j,Y'_j),(X_{i'},Y_{i'}),$ and $(X_{j'}',Y_{j'}')$. First, we claim that if one of these variables is independent of the rest, then $\E(\xi_{i,j} \xi_{i',j'}|\mathcal{D}_{n,-k}) = 0.$ For instance, suppose that $(X_i,Y_i)$ is independent of the rest. Then we can expand
\begin{align*}
    \E(\xi_{i,j} \xi_{i',j'}|\mathcal{D}_{n,-k}) &= \E(\E(\xi_{i,j}|\mathcal{D}_{n,-k},(X'_j,Y'_j),(X_{i'},Y_{i'}),(X_{j'}',Y_{j'}')) \xi_{i',j'}|\mathcal{D}_{n,-k})\\
    &=\E(\E(\xi_{i,j}|(X'_j,Y'_j),\mathcal{D}_{n,-k})\xi_{i',j'}|\mathcal{D}_{n,-k}).
\end{align*}
Now we can expand the inner expectation by
\begin{align*}
    &\E(\xi_{i,j}|(X'_j,Y'_j),\mathcal{D}_{n,-k})\\
    &= \E(R_{ij}(\mathcal{D}_{n,-k}) -\mu_{n,k}|(X'_j,Y'_j),\mathcal{D}_{n,-k}) - \E(\phi(X_i,Y_i;\mathcal{D}_{n,k})(\mathcal{D}_{n,-k}) -\mu_{n,k}|(X'_j,Y'_j),\mathcal{D}_{n,-k})\\
    &-\psi(X_j',Y_j';\mathcal{D}_{n,k})(\mathcal{D}_{n,-k}) -\mu_{n,k}.
\end{align*}
\begin{sloppypar}
Using that $\E(R_{ij}(\mathcal{D}_{n,-k}) -\mu_{n,k}|(X'_j,Y'_j),\mathcal{D}_{n,-k}) = \psi(X_j',Y_j';\mathcal{D}_{n,k})(\mathcal{D}_{n,-k}) -\mu_{n,k}$ and that $\E(\phi(X_i,Y_i;\mathcal{D}_{n,k})(\mathcal{D}_{n,-k}) -\mu_{n,k}|(X'_j,Y'_j),\mathcal{D}_{n,-k}) = 0$, we can conclude that $\E(\xi_{i,j} \xi_{i',j'}|\mathcal{D}_{n,-k}) = 0.$
\end{sloppypar}
If either 
\begin{enumerate}
    \item $i=i'$ or $i'=j$
    \item $j'=i$ or $j'=j$
\end{enumerate}
holds but not both, then one of $(X_i,Y_i),(X'_j,Y'_j),(X_{i'},Y_{i'}),$ and $(X_{j'}',Y_{j'}')$ is independent of the others so $\E(\xi_{i,j} \xi_{i',j'}|\mathcal{D}_{n,-k}) = 0.$ There are $O(n_k^2)$ terms where both conditions hold so there are at most $O(n_k^2)$ terms where $\E(\xi_{i,j} \xi_{i',j'}|\mathcal{D}_{n,-k}) \neq 0$. Denote the pairs of indicies which are nonzero by $\Omega.$ By Chebyshev, we have concluded the bound
\begin{align*}
    \P((*)\geq \epsilon) &\leq \frac{\E\left[\frac{n}{\sigma_n^2 n_k^4} \sum_{(i,j)(i',j')\in \Omega} \E(\xi_{i,j} \xi_{i',j'}|\mathcal{D}_{n,-k})\right]}{\epsilon^2}\\
    &= \frac{\frac{n}{\sigma_n^2 n_k^4} \sum_{(i,j)(i',j')\in \Omega} \E[\xi_{i,j} \xi_{i',j'}]}{\epsilon^2}\\
    &= O\left(\frac{1}{n_k}\right) \rightarrow 0.
\end{align*}
The final equality follows from the fact that the sum contains at most $O(n_k^2)$ terms each bounded above, and by Assumption \ref{as: crossfit}, the variance is bounded below by a constant.
\end{proof}

\begin{lemma}
\label{lem: stability u statistic}
    Under Assumption \ref{as: crossfit}, we have
    $$
        \norm{\nabla_i[ \phi(X,Y;\mathcal{D}_n) + \psi(X,Y';\mathcal{D}_n)]}_2 = o(n^{-1/2}).
    $$
\end{lemma}

\begin{proof}
    We need a stability bound
$$
 \norm{\nabla_i[\phi(X,Y;\mathcal{D}_n) + \psi(X,Y';\mathcal{D}_n)]}_2 = o(n^{-1/2}).
$$
By linearity, it is sufficient to bound $\norm{\nabla_i\phi(X,Y;\mathcal{D}_n)}_2$ and $\norm{\nabla_i\psi(X,Y';\mathcal{D}_n)}_2$ separately. We show the argument for $\norm{\nabla_i\phi(X,Y;\mathcal{D}_n)}_2$. The argument for $\norm{\nabla_i\psi(X,Y';\mathcal{D}_n)}_2$ is similar. Let $\mathcal{D}_n^i$ denote the data $\mathcal{D}_n$ where the $i$-th data point $Z_i:=(X_i,Y_i,Y'_i)$ is replaced by an iid copy $\tilde{Z}_i$. To simplify notation, we additionally define $\hat g_n := \wh{D}(\mathcal{D}_n)$ and $\hat g_n^i := \wh{D}(\mathcal{D}_n^i)$ to denote the output of the distinguisher with respect to the data $\mathcal{D}_n$ along with its replace one counterpart.

By definition, we can write
\begin{align*}
&\norm{\nabla_i\phi(X,Y;\mathcal{D}_n)}_2 \\
&= \norm{\phi(X,Y;\mathcal{D}_n)-\phi(X,Y;\mathcal{D}_n^i)}_2\\
&= \norm{\E_{X',Y'}[\mathds{1}(\hat g_n(X,Y)<\hat g_n(X',Y'))|\mathcal{D}_n]-\E_{X',Y'}[\mathds{1}(\hat g_n^i(X,Y)<\hat g_n^i(X',Y'))|\mathcal{D}_n^i]}_2\\
&=\norm{\E_{X',Y'}[\mathds{1}(\hat g_n(X,Y)<\hat g_n(X',Y'))-\mathds{1}(\hat g_n^i(X,Y)<\hat g_n^i(X',Y'))|\mathcal{D}_n,\tilde{Z}_i]}_2.
\end{align*}

We bound the inner term
\begin{align*}
&\left|\E_{X',Y'}[\mathds{1}(\hat g_n(X,Y)<\hat g_n(X',Y'))-\mathds{1}(\hat g_n^i(X,Y)<\hat g_n^i(X',Y'))|\mathcal{D}_n,\tilde{Z}_i]\right|\\
&\leq \E_{X',Y'}\left[\left|\mathds{1}(\hat g_n(X,Y)<\hat g_n(X',Y'))-\mathds{1}(\hat g_n^i(X,Y)<\hat g_n^i(X',Y'))\right|\bigg|\mathcal{D}_n,\tilde{Z}_i\right].
\end{align*}

Define
\begin{align*}
    \Delta &= \hat g_n(X,Y) - \hat g_n^i(X,Y)\\
    \Delta' &= \hat g_n(X',Y') - \hat g_n^i(X',Y').
\end{align*}

By similar argument as in \citet{huTwoSampleConditionalDistribution2024a}(proof of Proposition 1), we can write
\begin{align*}
    &\left|\mathds{1}(\hat g_n(X,Y)<\hat g_n(X',Y'))-\mathds{1}(\hat g_n^i(X,Y)<\hat g_n^i(X',Y')\right|\\
    &\leq \mathds{1}\left(\left|\hat g_n(X,Y)-\hat g_n(X',Y')\right| < |\Delta-\Delta'|\right).
\end{align*}

For any $\epsilon > 0$, we can bound 
\begin{align*}
    &\E_{X',Y'}\left[\mathds{1}\left(\left|\hat g_n(X,Y)-\hat g_n(X',Y')\right| < |\Delta-\Delta'|\right)\bigg|\mathcal{D}_n,\tilde{Z}_i\right]\\
    &\leq \P_{X',Y'}\left(\left|\hat g_n(X,Y)-\hat g_n(X',Y')\right| <\epsilon\bigg|\mathcal{D}_n,\tilde{Z}_i\right) + \P_{X',Y'}\left(|\Delta-\Delta'| > \epsilon\bigg|\mathcal{D}_n,\tilde{Z}_i\right).
\end{align*}

Under the bounded density assumption in Assumption \ref{as: crossfit}, the first term can be bounded by $2B\epsilon$. Condition on any realization of $\mathcal{D}_n,\tilde{Z}_i, X$ and $Y$. Under Assumption \ref{as: crossfit}, $\Delta'$ is $(\log(n)^{-\beta}n^{-1/2},\beta)-SW$, so the second term can be bounded up to a constant by $\exp\left(-(\epsilon\log(n)^{\beta}n^{1/2})^{1/\beta}\right),$ noting that after the conditioning $\Delta$ is deterministic so $\Delta-\Delta'$ is SW with same parameters. Then we can choose $\epsilon \asymp \log(n)^{-0.5\beta}n^{-1/2}$ to conclude.
\end{proof}

\begin{proof}[Proof of Proposition \ref{thm: crossfit asymp norm}]

    First we make the decomposition 
    \begin{align*}
        &\frac{\sqrt{n}}{\sigma_{tr,cross}} \left(T_{n,cross}-\frac{1}{K}\sum_{k=1}^K \mu(\mathcal{D}_{n,-k})\right)\\
        &= \frac{\sqrt{n}}{\sigma_{tr,cross}}\left(\frac{1}{K}\sum_{k=1}^K\left(\frac{1}{n_k^2}\sum_{i,j \in \mathcal{I}_k}R_{ij}(\mathcal{D}_{n,-k}) - \mu(\mathcal{D}_{n,-k})\right)\right)\\
        &= \frac{\sqrt{n}}{\sigma_{tr,cross}}\left(\frac{1}{K}\sum_{k=1}^K\left(\frac{1}{n_k^2}\sum_{i,j \in \mathcal{I}_k}R_{ij}(\mathcal{D}_{n,-k}) - \mu(\mathcal{D}_{n,-k})-\Phi_k-\Psi_k+\Phi_k+\Psi_k\right)\right)\\
        &= \underbrace{\frac{\sqrt{n}}{\sigma_{tr,cross}}\left(\frac{1}{K}\sum_{k=1}^K\frac{1}{n_k^2}\sum_{i,j \in \mathcal{I}_k}R_{ij}(\mathcal{D}_{n,-k}) - \mu(\mathcal{D}_{n,-k})-\Phi_k-\Psi_k\right)}_{(I)} + \underbrace{\frac{\sqrt{n}}{\sigma_{tr,cross}}\left(\frac{1}{K}\sum_{k=1}^K\Phi_k+\Psi_k\right)}_{(II)},
    \end{align*}
    where $\Phi_k$ and $\Psi_k$ are defined in Lemma \ref{lm: hajek decomposition}. Using Lemma \ref{lm: hajek decomposition}, we have that $(I) \xrightarrow{p} 0$. Term $(II)$ can be written as
    $$
    \frac{\sqrt{n}}{\sigma_{tr,cross}}\left(\frac{1}{K}\sum_{k=1}^K\Phi_k+\Psi_k\right) = \frac{\sqrt{n}}{\sigma_{tr,cross}}\left(\frac{1}{n}\sum_{k=1}^K\sum_{i\in \mathcal{I}_{n,k}}\phi(X_i,Y_i;\mathcal{D}_{n,-k}) + \psi(X_i,Y'_i;\mathcal{D}_{n,-k}) - 2\mu(\mathcal{D}_{n,-k})\right).
    $$
    The result follows from applying the CV-CLT Theorem \ref{thm: stability clt} noting that the stability assumption is satisfied by Lemma \ref{lem: stability u statistic}.
\end{proof}

\subsection{Sample Split Asymptotics}
In this section, we give the proof of Theorem $\ref{thm: asymp normality ss}.$ The strategy is to first prove asymptotic normality conditional on $\mathcal{D}_{n,1}$, the sample split used for fitting the distinguisher, and then extending to unconditional asymptotic normality. 

\begin{proof}[Proof of Proposition \ref{thm: asymp normality ss}]

Condition on $\mathcal{D}_{n,1}$. Recall the projections $\psi,\phi$ defined in (\ref{eq:projections}).
Similar to the cross-fit case, we start by decomposing the two sample U-statistic $T_{n,split}$.
Let
    \begin{align*}
        \tilde{\Phi} &= \frac{1}{n_2}\sum_{i \in \mathcal{I}_{2}}\phi(X_i,Y_i,\mathcal{D}_{n,1}) - \mu(\mathcal{D}_{n,1})\\
        \tilde{\Psi} &= \frac{1}{n_2}\sum_{i \in \mathcal{I}_{2}}\psi(X'_i,Y'_i,\mathcal{D}_{n,1}) - \mu(\mathcal{D}_{n,1}).
    \end{align*} 
    
    Then by Lemma \ref{lm: hajek decomposition},
    $$
        \sqrt{n_2}\left(T_{n,split} - \mu(\mathcal{D}_{n,1}) -  \tilde{\Phi} -  \tilde{\Psi}\right) \xrightarrow{p} 0.
    $$

Using Lemma \ref{lm: hajek decomposition}, $T_{n,split}-\mu(\mathcal{D}_{n,1})$ is asymptotically equivalent to $\tilde{\Phi}_k - \tilde{\Psi}_k = \frac{1}{n_2}\sum_{i \in \mathcal{I}_{2}}\phi(X_i,Y_i,\mathcal{D}_{n,1}) +\psi(X'_i,Y'_i,\mathcal{D}_{n,1}).$ Conditional on $\mathcal{D}_{n,1}$, $\tilde{\Phi}_k - \tilde{\Psi}_k$ is the sample mean of $n_2$ iid random variables $\phi(X_i,Y_i,\mathcal{D}_{n,1}) +\psi(X'_i,Y'_i,\mathcal{D}_{n,1}),$ the boundedness of $\phi,\psi$ imply that these additionally have finite absolute third moments. In particular, we have Berry-Esseen bound
$$
    \sup_{w\in \R}\left|\P\left(\frac{\sqrt{n_2}}{\sigma_{n,split}}[\tilde{\Phi} +\tilde{\Psi}] \leq w\bigg|\mathcal{D}_{n,1}\right) - \Upsilon(w)\right| \lesssim \frac{1}{\sigma_{n,split}^3\sqrt{n_2}},
$$
where $\Upsilon$ denotes the CDF of standard normal.

This says that for any $w\in \R$, we have
$$
\Upsilon(w) -  \frac{1}{\sigma_{n,split}^3\sqrt{n_2}} \lesssim \P\left(\frac{\sqrt{n_2}}{\sigma_{n,split}}[\tilde{\Phi} +\tilde{\Psi}] \leq w\bigg|\mathcal{D}_{n,1}\right) \lesssim \Upsilon(w) +  \frac{1}{\sigma_{n,split}^3\sqrt{n_2}}.
$$ Taking expectation over the fitting randomness in $\mathcal{D}_{n,1}$, we arrive at
$$
\Upsilon(w) -  \P\left(\frac{1}{\sigma_{n,split}^3\sqrt{n_2}}\right) \lesssim \P\left(\frac{\sqrt{n_2}}{\sigma_{n,split}}[\tilde{\Phi} +\tilde{\Psi}] \leq w\right) \lesssim \Upsilon(w) +  \P\left(\frac{1}{\sigma_{n,split}^3\sqrt{n_2}}\right).
$$
Similar to the proof of Theorem 3 in \citet{caiAsymptoticDistributionFreeIndependence2024a}, we can show $\P\left(\frac{1}{\sigma_{split}^3\sqrt{n_2}}\right) = o(1)$, as $n_1,n_2\rightarrow\infty$ under the growth condition on $\sigma_{split}$ in Assumption \ref{as: ss}. This concludes that 
$$
    \sup_{w\in \R}\left|\P\left(\frac{\sqrt{n_2}}{\sigma_{n,split}}[\tilde{\Phi} +\tilde{\Psi}] \leq w\right) - \Upsilon(w)\right| \lesssim o(1),
$$
completing the proof.
\end{proof}

\subsection{Variance Estimation}
In this section we give the proof of Proposition \ref{prop: var est split} and Proposition \ref{prop: var est cross} on variance estimation. For brevity, we give the proof of Proposition \ref{prop: var est cross} in detail and note that the proof of Proposition \ref{prop: var est split} follows from a modified argument in step 2 involving only a single split.
\begin{proof}[Proof of Proposition \ref{prop: var est cross}]
We first introduce an intermediate variance term
$$
    \tilde{\sigma}^2_{tr,cross} := \frac{1}{K}\sum_{k=1}^K \frac{1}{n_k}\sum_{i\in \mathcal{I}_{n,k}} \left(\phi(X_i,Y_i;\mathcal{D}_{n,-k})+\psi(X_i,Y'_i;\mathcal{D}_{n,-k}) - 2T_{n,-k}\right)^2.
$$ where we recall that
$$
    T_{n,-k} = \frac{1}{n_k^2}\sum_{i,j \in \mathcal{I}_k}R_{ij}(\mathcal{D}_{n,-k}).
$$
The argument proceeds using the following two steps:
\begin{enumerate}
    \item Show that $\tilde{\sigma}^2_{tr,cross}$ is a good approximation of $ \sigma^2_{tr,cross}$. The major difference between the two quantities is that the variance term $\tilde{\sigma}^2_{tr,cross}$ has additional randomness in the fitting folds $\mathcal{D}_{n,-k}$. The key observation is that by Efron-Stein inequality and the stability assumption (\ref{as: crossfit}), this additional fitting randomness is negligible.
    \item Show that $\tilde{\sigma}^2_{tr,cross}$ is close to $\hat{\sigma}^2_{tr,cross}$. In this step, we control the differences between the term $\tilde{\sigma}^2_{tr,cross}$ which involve the true projections $\phi,\psi$ and the term $\hat{\sigma}^2_{tr,cross}$ which replaces them by their empirical versions $\hat{\phi}$ and $\hat{\psi}$.
\end{enumerate}
\noindent\textbf{Step 1:} By similar argument as in the proof of Theorem 4.6 in \citet{leiModernTheoryCrossValidation2025}, and that $\phi,\psi$ are uniformly bounded by $1$, we can conclude that
$$
    \frac{\tilde{\sigma}_{tr,cross}^2}{\sigma_{tr,cross}^2} \xrightarrow{P} 1.
$$
\noindent\textbf{Step 2:} Consider fold $k \in [K]$. We first aim to bound $|\tilde{\sigma}_{tr,cross}^2 - \hat{\sigma}_{tr,cross}^2|$ by
\begin{align*}
    \bigg|\bigg(\phi(&X_i,Y_i;\mathcal{D}_{n,-k})+\psi(X_i,Y'_i;\mathcal{D}_{n,-k})\bigg)^2 - \bigg(\hat{\phi}(X_i,Y_i;\mathcal{D}_{n,-k})+\hat{\psi}(X_i,Y'_i;\mathcal{D}_{n,-k})\bigg)^2 \bigg|\\
    &\leq \underbrace{\left|\phi(X_i,Y_i;\mathcal{D}_{n,-k})^2-\hat{\phi}(X_i,Y_i;\mathcal{D}_{n,-k})^2\right|}_{(I)}\\
    &+ \underbrace{\left|\psi(X_i,Y'_i;\mathcal{D}_{n,-k})^2-\hat{\psi}(X_i,Y'_i;\mathcal{D}_{n,-k})^2\right|}_{(II)}\\
    &+ \underbrace{\left|\phi(X_i,Y_i;\mathcal{D}_{n,-k})\psi(X_i,Y'_i;\mathcal{D}_{n,-k})-\hat{\phi}(X_i,Y_i;\mathcal{D}_{n,-k})\hat{\psi}(X_i,Y'_i;\mathcal{D}_{n,-k})\right|}_{(III)}.
\end{align*}
We can bound these terms separately. We will show the claim for term $(II)$ in detail. The rest can be bounded in a similar manner. We can expand
\begin{align*}
    (II) &\leq |\psi(X_i,Y'_i;\mathcal{D}_{n,-k})-\hat{\psi}(X_i,Y'_i;\mathcal{D}_{n,-k})||\psi(X_i,Y'_i;\mathcal{D}_{n,-k})|\\
    &+ |\psi(X_i,Y'_i;\mathcal{D}_{n,-k})-\hat{\psi}(X_i,Y'_i;\mathcal{D}_{n,-k})||\hat{\psi}(X_i,Y'_i;\mathcal{D}_{n,-k})|.
\end{align*}

Notice that $\psi(X',Y';\mathcal{D}_{n,-k})$ is the CDF of the random variable $\hat g_{n,-k}(X,Y)$ evaluated at $\hat g_{n,-k}(X',Y')$. Similarly, $\hat{\psi}(X',Y';\mathcal{D}_{n,-k})$ is the empirical CDF of $\hat g_{n,-k}(X,Y)$ evaluated at $\hat g_{n,-k}(X',Y')$. By DKW inequality, the difference is bounded by $O(n_{k}^{-1/2})$. Using that $|\hat{\psi}(X_i,Y'_i;\mathcal{D}_{n,-k})|$ and $|\psi(X_i,Y'_i;\mathcal{D}_{n,-k})|$ are both bounded by $1$, we conclude that $(II) = O_p(n_{k}^{-1/2})$. Similar arguments show that $(I) = O_p(n_{k}^{-1/2})$ and $(III) = O_p(n_{k}^{-1/2})$ as well. Thus,
$$
\bigg|\bigg(\phi(X_i,Y_i;\mathcal{D}_{n,-k})+\psi(X_i,Y'_i;\mathcal{D}_{n,-k})\bigg)^2 - \bigg(\hat{\phi}(X_i,Y_i;\mathcal{D}_{n,-k})+\hat{\psi}(X_i,Y'_i;\mathcal{D}_{n,-k})\bigg)^2 \bigg| = O_p(n_{k}^{-1/2}).
$$
Combining across all folds shows that $$\tilde{\sigma}_{tr,cross}^2 -\hat{\sigma}_{tr,cross}^2 = O_p(n^{-1/2}).$$
Now we can bound
$$
    \frac{\hat{\sigma}_{tr,cross}^2}{\sigma_{tr,cross}^2} = \frac{(\hat{\sigma}_{tr,cross}^2 - \tilde{\sigma}_{tr,cross}^2)+\tilde{\sigma}_{tr,cross}^2}{\sigma_{tr,cross}^2} = O_p(n^{-1/2}) + \frac{\tilde{\sigma}_{tr,cross}^2}{\sigma_{tr,cross}^2} \xrightarrow{P} 1,
$$
where we additionally use that $\sigma_{tr,cross}^2$ is lower bounded by a constant given in Assumption \ref{as: crossfit}.
\end{proof}

\subsection{Power}
\begin{proof}[Proof of Corollary \ref{cor:power}]
    By the asymptotic normality of the test statistic we know that
    \begin{align*}
        &\frac{\sqrt{n}}{\hat{\sigma}_{tr,cross}}\Big(T_{n,cross} - \delta_n - \frac{1}{2}\Big)\\
        &= Z + \frac{\sqrt{n}}{\hat{\sigma}_{tr,cross}}\Big(\frac{1}{K}\sum_{k=1}^K \AUC(\wh{L}_k) - \delta_n - \frac{1}{2}\Big) + o_p(1)\\
        &\geq Z + \frac{\sqrt{n}}{\hat{\sigma}_{tr,cross}}\Big(\frac{1}{K}\sum_{k=1}^K (\AUC(\wh{L}_k) - \AUC(L)) + \delta^* - \delta_n \Big) + o_p(1)\\
        &\geq Z + \underbrace{\frac{\hat{\sigma}_{tr,cross}}{\sigma_{tr,cross}}\frac{\sqrt{n}}{c}\Big(\frac{1}{K}\sum_{k=1}^K (\AUC(\wh{L}_k) - \AUC(L)) + \delta^* - \delta_n \Big) \mathds{1}\Big\{\norm{\wh{L}_k-L} \leq r_{n} \forall k \in [K]\Big\}}_{(I)}\\
        &\quad+ \underbrace{\frac{\hat{\sigma}_{tr,cross}}{\sigma_{tr,cross}}\frac{\sqrt{n}}{c}\Big(\frac{1}{K}\sum_{k=1}^K (\AUC(\wh{L}_k) - \AUC(L)) + \delta^* - \delta_n \Big) \mathds{1}\Big\{\exists k \in [K]:\norm{\wh{L}_k-L} > r_{n} \Big\}}_{(II)} + o_p(1).
    \end{align*}
    Here $Z$ is a standard normal variable, $c$ is the deterministic variance lower bound in Assumption \ref{as: crossfit}, and to get the second line we use $\frac{1}{2} = \AUC(L) - \delta^*$. We now need to bound the terms $(I),(II)$. 

    For term $(I)$, using Proposition \ref{prop:auc_cons}, and the condition $\delta^*-\delta_n - 2r_{n} = \omega(1/\sqrt{n})$, we know that
    \begin{align*}
        \frac{\sqrt{n}}{c}\Big(\frac{1}{K}\sum_{k=1}^K (\AUC(\wh{L}_k) - \AUC(L)) + \delta^* - \delta_n \Big)
        \geq \frac{\sqrt{n}}{c}\Big(-2r_{n} + \delta^* - \delta_n \Big) = \omega(1).
    \end{align*}
    Using Proposition \ref{prop: var est cross} and that $\P\Big(\norm{\wh{L}_k-L} \leq r_{n} \forall k \in [K]\Big) \rightarrow 1$, we can conclude that $(I)$ diverges in probability.

    For term $(II)$, using Proposition \ref{prop: var est cross}, the fact that $\Big(\frac{1}{K}\sum_{k=1}^K (\AUC(\wh{L}_k) - \AUC(L)) + \delta^* - \delta_n \Big)$ is bounded and that $\P\Big(\exists k \in [K]:\norm{\wh{L}_k-L} > r_{n}\Big) \rightarrow 0$, we can conclude that $(II) \rightarrow 0$ in probability.

    Combining the control on term $(I),(II)$, we can conclude that the cross-fit test statistic diverges in probability.
\end{proof}
\subsection{Classification with Coupled Samples}

In this section we prove Theorem \ref{thm:consistency}. By standard LASSO analysis \citet{wainwrightHighDimensionalStatisticsNonAsymptotic2019} (Theorem 7.13(a)) we can get $\ell_2$ control of the solution $\hat{\beta}$ to $(\ref{eq: lasso problem})$. In what follows, we define
$
    w = Z- \Xi^\top \beta^*,
$
to denote the residuals.

\begin{theorem}
\label{thm: LASSO bound}
    Under Assumptions \ref{as: sparse basis rep} and \ref{as: LASSO}, the estimate $\hat{\beta}$ satisfies
    $
        \norm{\hat{\beta}-\beta^*}_2 \leq \frac{3}{\kappa} \sqrt{s}\lambda,
    $
    when $\lambda \geq 4\norm{\frac{\Xi^\top w}{n}}_\infty$.
\end{theorem}

\begin{lemma}
    \label{lm: lambda bound}
    Suppose that each basis function $e_i,i=1,...,K_n$ is bounded by $[-B,B]$. For some constant $C$, we have
    $$
    P\left(\norm{\frac{\Xi^\top w}{n}}_\infty\geq C\sqrt{\frac{\log K_n}{n}}\right) \leq {K_n}^{-1}.
    $$
\end{lemma}

\begin{proof}
First let's consider a single component of $\Xi^\top w$. The $k$-th component is given by
\begin{equation*}
    \zeta_k := \begin{pmatrix}
        e_k(X_1,Y_1)\\
        \cdots\\
        e_k(X_n,Y_n)\\
        e_k(X'_1,Y'_1)\\
        \cdots\\
        e_k(X'_n,Y'_n)
    \end{pmatrix}^\top w.
\end{equation*}
This entry can be viewed as a function of random variables $(X_1,Y_1),...,(X_n,Y_n),(X'_1,Y'_1),...,(X'_n,Y'_n)$. Suppose we replace a single one of these random variables by an iid copy. Without loss of generality, we can replace $(X_1,Y_1)$ by $(\tilde{X}_1,\tilde{Y_1})$. By our dependency structure, $X_1'$ is also replaced by $\tilde{X}_1$. After replacement by the iid copy, the component is now 
\begin{equation*}
   \tilde{\zeta}_k:=
   \begin{pmatrix}
        e_k(\tilde{X}_1,\tilde{Y}_1)\\
        \cdots\\
        e_k(X_n,Y_n)\\
        e_k(\tilde{X}_1,Y'_1)\\
        \cdots\\
        e_k(X'_n,Y'_n)
    \end{pmatrix}^\top w.
\end{equation*}
Using the boundness of $e_k$, we see that
$$
    |\zeta_k-\tilde{\zeta}_k| \leq 4B.
$$
Thus, $\zeta_k$ satisfies the bounded differences condition so we can apply McDiarmids inequality to get
$$
        P(|\zeta_k| \geq t) \leq 2\exp\left(\frac{-nt^2}{8B^2}\right).
$$
This bound holds across all components, so we can take a union bound to get
    $$
        P\Big(\norm{\Xi^\top w/n}_\infty \geq t\Big) \leq 2K_n\exp\left(\frac{-nt^2}{8B^2}\right).
    $$
    The result follows by taking
    $$
        t \asymp \sqrt{\log K_n/n}.
    $$
\end{proof}

The proof of Theorem \ref{thm:consistency} is a simple corollary of Lemma \ref{lm: lambda bound}.

\begin{proof}[Proof of Theorem \ref{thm:consistency}]
    By Lemma \ref{lm: lambda bound}, with probability at least $1-{K_n}^{-1}$, we can take $\lambda \asymp \sqrt{(\log K_n)/n}$.
\end{proof}

\label{sec: Proofs}

\begin{singlespace}
\medskip
\bibliographystyle{plainnat}
\bibliography{gof.bib}
\end{singlespace}

\end{document}